\UseRawInputEncoding
\documentclass[final]{siamart220329}
\usepackage{amsfonts}
\usepackage{amssymb}
\usepackage{mathrsfs}
\usepackage{upgreek}
\usepackage{graphicx,epstopdf} 
\usepackage[caption=false]{subfig}
\usepackage{pgfplots}
\usepackage{algorithmicx}
\usepackage{algpseudocode}
\usepackage{array} 
\usepackage{tabularx}
\usepackage{booktabs}
\usepackage{bm}
\usepackage{multirow}
\usepackage{multicol}
\usepackage[normalem]{ulem}
\usepackage{cancel}
\pgfplotsset{compat=1.17}
\newsiamremark{remark}{Remark}

\newcommand{\Z}{\mathbb{Z}}
\newcommand{\N}{\mathbb{N}}
\newcommand{\R}{\mathbb{R}}
\newcommand{\Le}{\mathscr{L}}
\newcommand{\dof}{\mathcal{M}}
\newcommand{\ind}{\nu}
\newcommand{\aceind}{\pmb{\nu}}

\newcommand{\st}{\sigma}
\newcommand\xx{\pmb{x}}
\newcommand{\Pa}{\uppi}
\newcommand\Hs{\mathscr{X}}
\newcommand\ham{\mathcal{H}}
\newcommand{\<}{\langle}
\renewcommand{\>}{\rangle}
\newcommand\dd{~{\rm d}}
\newcommand\El{E_{\rm L}}
\newcommand\EE{\mathbb{E}}
\newcommand\eb{\mathcal{T}}
\newcommand{\basis}{\phi}
\newcommand{\nus}{N_{\uparrow}}
\newcommand{\nds}{N_{\downarrow}}
\newcommand{\us}{\uparrow}
\newcommand{\ds}{\downarrow}
\newcommand{\ace}{\pmb{A}}
\newcommand{\aceD}{\mathcal{D}}
\newcommand{\aceB}{\mathcal{B}}
\newcommand{\aceI}{\mathcal{I}}
\newcommand{\param}{\pmb{c}}
\newcommand{\params}{\pmb{c}}
\newcommand{\wfs}{\Psi_{\rm s}}
\newcommand\rr{\mathbf{r}}
\newcommand{\wfv}{\Psi_{\rm V}}
\newcommand{\wfbf}{\Psi_{\rm BF}}
\newcommand\loss{\mathcal{L}}
\newcommand{\grad}{\mathbb{g}}
\newcommand{\sse}{\mathit{S}}
\newcommand{\prj}{\varPi}
\renewcommand{\le}{L}
\newcommand{\trans}{\top}

\title{A Multilevel Method for Many-Electron Schr\"{o}dinger Equations Based on the Atomic Cluster Expansion\thanks{\today
\funding{DZ and HC's work was supported by the National Key R\&D Program of China (No. 2020YFA0712900) and National Natural Science Foundation of China (No. NSFC11971066)}.
}
}

\author{Dexuan Zhou\thanks{School of Mathematical Sciences, Beijing Normal University, Beijing 100875 China
       (\email{zhoudexuan@mail.bnu.edu.cn}).}
\and Huajie Chen\thanks{School of Mathematical Sciences, Beijing Normal University, Beijing 100875 China
       (\email{chen.huajie@bnu.edu.cn}).}
\and Cheuk Hin Ho\thanks{Department of Mathematics, University of British Columbia, 1984 Mathematics Road, Vancouver, BC, Canada V6T 1Z2 
       (\email{jerryho528@math.ubc.ca}).}
\and Christoph Ortner\thanks{Department of Mathematics, University of British Columbia, 1984 Mathematics Road, Vancouver, BC, Canada V6T 1Z2 
      (\email{ortner@math.ubc.ca}).}
}

\headers{A Multilevel Method for Many-Electron Schr\"{o}dinger Equations}{D. Zhou, H. Chen, C. Hin Ho and C. Ortner}
 
\begin{document}
\maketitle

\begin{abstract}
   The atomic cluster expansion (ACE) {\it (Drautz, 2019)} yields a highly efficient and intepretable parameterisation of symmetric polynomials that has achieved great success in modelling properties of many-particle systems.
   In the present work we extend the practical applicability of the ACE framework to the computation of many-electron wave functions.   
   To that end, we develop a customized variational Monte-Carlo algorithm that exploits the sparsity and hierarchical properties of ACE wave functions. 
   We demonstrate the feasibility on a range of proof-of-concept applications to one-dimensional systems.
\end{abstract}



\section{Introduction}
\label{sec:introduction}
Computing the ground state of the many-body Schr\"{o}dinger equation is arguably the most fundamental problem in electronic structure calculations. 
At the same time, it is also one of the most challenging computational tasks due to two reasons: First, the unknown wave function is a function of all electron coordinates and due to the electron-electron interaction does not normally factorize. 
Secondly, the required level of accuracy to obtain meaningful information about the system of interest is usually very high. 
The leading approaches used in practice include varying compromises between cost and accuracy: 
(i) Density Functional Theory \cite{Martin2005} is able to treat thousands of particles but provides insufficient accuracy for many systems of interest, especially highly correlated systems;
(ii) Post-Hartree-Fock methods \cite{helgaker2014molecular} provide accurate descriptions of quantum systems, but the ``gold standard" methods CCSD(T) and MP2(4) methods remain prohibitively expansive for large systems; 
(iii) quantum Monte-Carlo methods \cite{becca2017qmc} obtain approximations of the ground states by means of stochastic algorithms, and are the focus of the present work. Specifically we will introduce a new variant of the variational Monte Carlo (VMC) method, employing an atomic cluster expansion~\cite{drautz2019ACE,dusson2022ACE}
for the wave function parameterisation, as suggested in \cite{drautz2022}, and a cascadic multilevel approach for the ground state computation. 

The VMC method expresses a quantum expectation value in terms of a sequence of configurations distributed according to the distribution probability for a trial wave function. The ground state solutions can be obtained by optimizing the trial wave functions to minimize the energy expectation.
The accuracy and efficiency of the VMC approach hinges on the form of the wave function ansatz, the efficiency of the sampling and parameter optimisation algorithms. The most commonly used ansatz in VMC are Jastrow-Slater wave functions \cite{szabo2012modernqc}, which have been successful in many applications, despite their accuracy being fundamentally limited.

In recent years, the parameterisation of wave functions in terms of neural network architectures has vastly expanded the design space for wave functions accessible to VMC algorithms~\cite{han2019dnn,hermann2020dnn,pfau2020Ferminet,Scherbela22}. 
Employing deep neural network architectures has resulted in far more parameters than previous works, but this has led to unprecedented accuracy among VMC approaches; see for instance the NetKet~\cite{choo2020fermionic,giuse16}, PauliNet~\cite{hermann2020dnn} and FermiNet~\cite{pfau2020Ferminet} architectures.
The success of those methods can be ascribed to the expressive power of neural networks as well as the choice of the anti-symmetric ansatz employed in those works.

The present work builds on the idea of employing machine learning methodologies, but instead of deep neural networks we exploit ideas from the atomic cluster expansion (ACE) model \cite{drautz2019ACE,dusson2022ACE} first proposed for wave functions in \cite{drautz2022}. 
In essence, this ansatz replaces the $S(N)$-equivariant deep neural networks employed in PauliNet~\cite{hermann2020dnn} with a highly sparsified symmetric polynomial. 
The ACE model was originally developed for the parameterisation of interatomic potentials where it results in accurate and highly efficient models, competing with state-of-the-art machine learning techniques \cite{lysogorskiy2021PACE}. 
As a sparse polynomial expansion, the ACE model has the advantage over neural network models that there are clearly interpretable approximation parameters: the choice of basis, the polynomial degrees and the sparsification pattern. The polynomial degree enables a theoretical and numerical study of convergence, the sparsification pattern appears to be related to correlation in the many-electron system. 
We propose and study two concrete parameterisations within the larger design space of \cite{drautz2022} based, respectively, on the Vandermonde ansatz and the backflow ansatz.
Both of those wave function models have the interesting property that they are hierarchical. Our second contribution in this work is to exploit this model hierarchy and design a cascadic multilevel VMC algorithm, which significantly accelerates and stabilizes the parameter optimization process.

In the present work we focus on demonstrating the feasibility, {\it in principle}, of using ACE parameterisations by developing a robust and performant multi-level VMC algorithm.
To that end we will restrict numerical tests to a highly simplified setting of one-dimensional electrons and soft Coulomb interactions. The application to realistic three-dimensional systems will be studied in a separate work. Because of this restriction we will also limit our description of the wave function parameterisations to one-dimensional particles but note that, at least conceptually, the generalisation to arbitrary dimension is immediate~\cite{drautz2022}.

\vskip 0.2cm
{\bf Outline.}
The remainder of this paper is organized as follows: 
In \cref{sec:schrodinger}, we recall some background about the many-electron Schr\"{o}dinger equation and the VMC method.
In \cref{sec:wfAnsatz}, we review the ACE model and then use it to construct two anti-symmetric wave function parameterisations employing, respectively, the Vandermonde and backflow transformations. In \cref{sec:multilevelVMC}, 
we review the VMC method and introduce a cascadic multilevel VMC algorithm.
In \cref{sec:numerics}, we present numerical experiments on a range of model problems.

\vskip 0.2cm
{\bf Notation.}
Throughout this paper, we will denote the cardinality of a set $A$ by $\#A$.
Let $\Z_m:=\{\sigma_1,\cdots\sigma_m\}$ be a discrete set with cardinality $m$, which is endowed with an order $\sigma_1<\cdots<\sigma_m$.
In particular, we will denote the set of spin quantum numbers by $\Z_2 := \{\uparrow, \downarrow\}$,
endowed with the order $\downarrow < \uparrow$.
For any pair $\ind_1 = (k_1, s_1)$, $\ind_2 = (k_2, s_2) \in \N\times \Z_m$, we define $\ind_1\leq \ind_2$ by the lexicographic order, i.e., $\ind_1<\ind_2$ if and only if $k_1<k_2$ or $k_1=k_2,s_1<s_2$.
For $n\in\N_+$, we will then denote by $(\N \times \Z_m)_{\rm ord}^n$ the set of all ordered $n$-tuples, i.e., 
\begin{align}
\label{nord}
(\N \times \Z_m)_{\rm ord}^n 
:= \Big\{\aceind = (\ind_1, \cdots, \ind_n) \in (\N \times \Z_m)^n :~ \ind_1\leq \ind_2\leq \cdots \leq \ind_n \Big\}. 
\end{align}

\section{The Many-body Schr\"{o}dinger Equation}
\label{sec:schrodinger} \setcounter{equation}{0} \setcounter{figure}{0}
We consider one-dimensional many-particle systems with $M$ nuclei and $N$ electrons.
Let $r_i\in\R,~\st_i\in\Z_2~(i=1,\cdots,N)$ denote the position and spin coordinates of the electrons.
The electron state of the system is described by the many-electron wave function $\Psi(\xx)$ with the $N$-electron configuration $\xx=(x_1,\cdots,x_N)$ and $x_i=(r_i,\sigma_i)\in \R\times \Z_2$.
The wave function $\Psi$ is required to satisfy the {\it anti-symmetry} condition 
\begin{eqnarray}
\label{anti}
\Psi(x_{\Pa(1)},\cdots,x_{\Pa(N)}) = (-1)^{\epsilon_\Pa}\Psi(x_1,\cdots,x_N)
\end{eqnarray}
for any permutation $\Pa$, where $\epsilon_\Pa$ is the parity of $\Pa$.
If there is no sign $(-1)^{\epsilon_\Pa}$ in \eqref{anti}, then it is called the {\it symmetry} condition.
We will denote the admissible class for the $N$-electron wave functions by 
\begin{align*}
\Hs_N := \Big\{ \Psi\in L^2\big((\R\times\Z_2)^N,\R\big) : 
\nabla_{r_i}\Psi\in L^2\big((\R\times\Z_2)^N,\R\big),
\Psi~\text{is anti-symmetric} \Big\} ,
\end{align*}
where we have restricted ourselves to considering real-valued wave functions. 
Since we will minimize a {\it Rayleigh quotient} to obtain the ground state wave function (see \eqref{E:RQ}) we have ignored the normalization constraint $\|\Psi\|_{L^2}=1$.

Let $R_I\in\R$ and $Z_I\in\Z_+~(I=1,\cdots, M)$ denote, respectively, the positions and atomic numbers of the nuclei.
Under the Born-Oppenheimer approximation, the non-relativistic Hamiltonian for the electron system is a Hermitian operator on $\Hs_N$, defined by 
\begin{eqnarray}
\label{hamiltonian:SE}
\ham := -\frac{1}{2}\sum_{i=1}^N\nabla^2_{r_i} + \sum_{i=1}^{N} v_\mathrm{ext}(r_i) + \sum_{1\leq i<j\leq N}v_\mathrm{ee}\big(|r_i-r_j|\big) ,
\end{eqnarray}
where $v_\mathrm{ext}:\R\rightarrow\R$ represents the attraction from the nuclei, and $v_\mathrm{ee}:\R\rightarrow\R$ represents the repulsion between the electrons. 
Throughout this work, we consider the soft-Coulomb interaction \cite{coe2011soft,eberly1989soft,thiele2008soft} instead of the bare Coulomb interaction,
where electron-electron repulsion and nuclear-electron attraction take the forms
\begin{eqnarray}
\label{soft-Coulomb}
v_\mathrm{ee}(u) = \frac{1}{\sqrt{1+u^2}}
\qquad{\rm and}\qquad
v_\mathrm{ext}(r) = -\sum_{I=1}^M Z_I v_\mathrm{ee}(|r-R_I|) .
\end{eqnarray}
Within this model, the ground state wave functions do not exhibit cusps present in 3d Coulomb systems. 
However, the challenges arising from the long-range interactions and strong correlation in 3d Coulomb systems remains for these 1d model systems as well \cite{stoudenmire2012one,wagner2012reference}.

The electron ground state $\Psi_0$ of the system can be determined by solving the many-electron time-independent Schr\"odinger equation
\begin{eqnarray}
\label{schrodinger}
\ham \Psi_0 = E_0 \Psi_0 , 
\end{eqnarray}
where the ground state energy $E_0$ is the lowest eigenvalue of the Hamiltonian such that the corresponding eigenfunction $\Psi_0\in\Hs_N$.
The ground state solution can also be obtained by minimizing the energy functional (or, Rayleigh quotient) with respect to the admissible wave functions
\begin{equation}
\label{E:RQ}
E_0 = \min_{\Psi\in \Hs_N} E(\Psi) 
\qquad{\rm with}\quad
E(\Psi) = \frac{\<\Psi|\ham|\Psi\>}{\<\Psi|\Psi\>}
= \frac{\int\Psi(\xx)\big(\ham\Psi\big)(\xx)\dd\xx}{\int|\Psi(\xx)|^2\dd\xx} ,
\end{equation}
where we have used Dirac's bracket and the integral with respect to the electron configuration $\xx$ means
$\displaystyle\int_{\xx\in \big(\R\times\Z_2\big)^N} := \sum_{\sigma_1\in\Z_2}\cdots\sum_{\sigma_N\in\Z_2} \int_{\R}\dd r_1\cdots\int_{\R} \dd r_N$.

In variational Monte Carlo (VMC) methods \cite{becca2017qmc}, the energy functional is reformulated as
\begin{align}
\label{E:vmc}
E(\Psi) = \int \frac{\big(\ham\Psi\big)(\xx)}{\Psi(\xx)} \frac{|\Psi(\xx)|^2}{\<\Psi|\Psi\>} \dd\xx
=: \int \El(\xx;\Psi) P(\xx;\Psi) \dd\xx 
= \EE_{\xx\sim P(\cdot;\Psi)} \big[ \El(\xx;\Psi) \big] ,
\end{align}
where the probability density $P(\xx;\Psi)$ and local energy $\El(\xx;\Psi)$ are, respectively, given by
\begin{align}
\label{vmc:P}
P(\xx;\Psi) := \frac{\big|\Psi(\xx)\big|^2}{\int\big|\Psi(\xx)\big|^2\dd\xx}
\qquad\qquad{\rm and}
\qquad\qquad\qquad\qquad\qquad\qquad\qquad
\\[1ex]
\nonumber
\El(\xx;\Psi) := \frac{\big(\ham\Psi\big)(\xx)}{\Psi(\xx)}
= -\frac{1}{2}\sum_{i=1}^N \frac{\nabla^2_{r_i}\Psi(\xx)}{\Psi(\xx)} + \sum_{i=1}^{N} v_\mathrm{ext}(r_i) + \sum_{1\leq i<j\leq N}v_\mathrm{ee}(|r_i-r_j|) .
\end{align}
The VMC formulation \eqref{E:vmc} allows us to estimate the energy functional in \eqref{E:RQ} (a high dimensional integral) by employing Markov-Chain Monte Carlo sampling.

\section{Wave function representations}
\label{sec:wfAnsatz} \setcounter{equation}{0} \setcounter{figure}{0}
There are a wide variety of choices available to parameterize wave functions, typically via transformations of $S(N)$ equivariant tensors such as the classic backflow ansatz that was recently revived in \cite{pfau2020Ferminet}. 
In this section we review and specialize a design space for wave functions where the equivariant tensors are constructed in terms of the atomic cluster expansion (ACE) \cite{bachmayr2021polynomialACE,dusson2022ACE,lysogorskiy2021PACE,van2020regularisedACE}. In the ACE framework, symmetric functions of many particles are first expanded in a standard tensor product basis, which is then heavily sparsified in a physically-motivated manner. Moreover, the symmetries are effectively exploited to obtain a highly efficient evaluation algorithm. 
We narrow down that vast ACE wave function design space proposed in~\cite{drautz2022} to two concrete practical architectures that we then implement and test in \cref{sec:numerics} on a range of preliminary benchmarks. Our main method of interest is based on the backflow transformation, but for comparison we will also introduce a simpler variant employing Vandermonde determinants.

\subsection{One-particle basis}
\label{sec:onebody}
Let $\Omega:=\R\times\Z_m$ denote an abstract one-particle configuration space, with $\R$ the space for spatial coordinates and $\Z_m$ the set for the ``spin''. 
We will expand many-body functions in terms of tensor products of a one-particle basis, i.e. a basis $\{\phi_\nu\}$ for functions $f: \Omega \to \R$. 
In general, there is significant freedom in the selection of $\phi_\nu$. Since we are focusing on one-dimensional particles and smooth wave functions, we choose the one-particle basis functions $\basis_\ind$ given by 
\begin{align}
\label{basis}    
\basis_\ind(x) :=
\basis_{(k,s)}(x) :=
\Le_k\bigg(\frac{2}{\pi}\arctan\Big(\frac{r}{L}\Big)\bigg)\cdot\chi_s(\sigma)
\quad{\rm for}~ x=(r,\sigma)\in\R\times\Z_m ,
\end{align}
where the index is a tuple $\ind=(k,s)$ with $k\in\N$ and $s\in\Z_m$,
$\Le_k$ are the standard Legendre polynomials, and $L>0$ is a length scaling constant depending on the size of region in which the atoms are spread. 
Finally, $\chi_s(\sigma)= \delta_{s \sigma}$ is a ``one-hot embedding'', that is,  $\{\chi_s\}_{s\in\Z_m}$ is a complete basis over the discrete set $\Z_m$.
The coordinate transform $r \mapsto \frac{2}{\pi}\arctan(r/L)$ transforms the space $\R$ to the domain $[-1, 1]$ of the Legendre polynomials. 
As a result of this transform, the resolution of the basis will be highest near the origin, where the atomic nuclei will be located and thus where resolution is required. 

The many-body wave functions we construct later will also be modulated by an envelope of the form $\exp(-\theta\cdot \sqrt{1+r^2})$, with $\theta$ a parameter, to ensure decay at infinity. 
To demonstrate the effect of the transform and envelope we plot the first six basis functions in \cref{fig:Tk}.
We visually confirm that this construction generates basis functions that oscillate in the ``core region'' and decay smoothly as $r\rightarrow\infty$, thus mimicking atomic orbitals used in electronic structure calculations~\cite{Martin2005}.

\begin{figure}[tbhp]
\centering
\includegraphics[width=12cm]{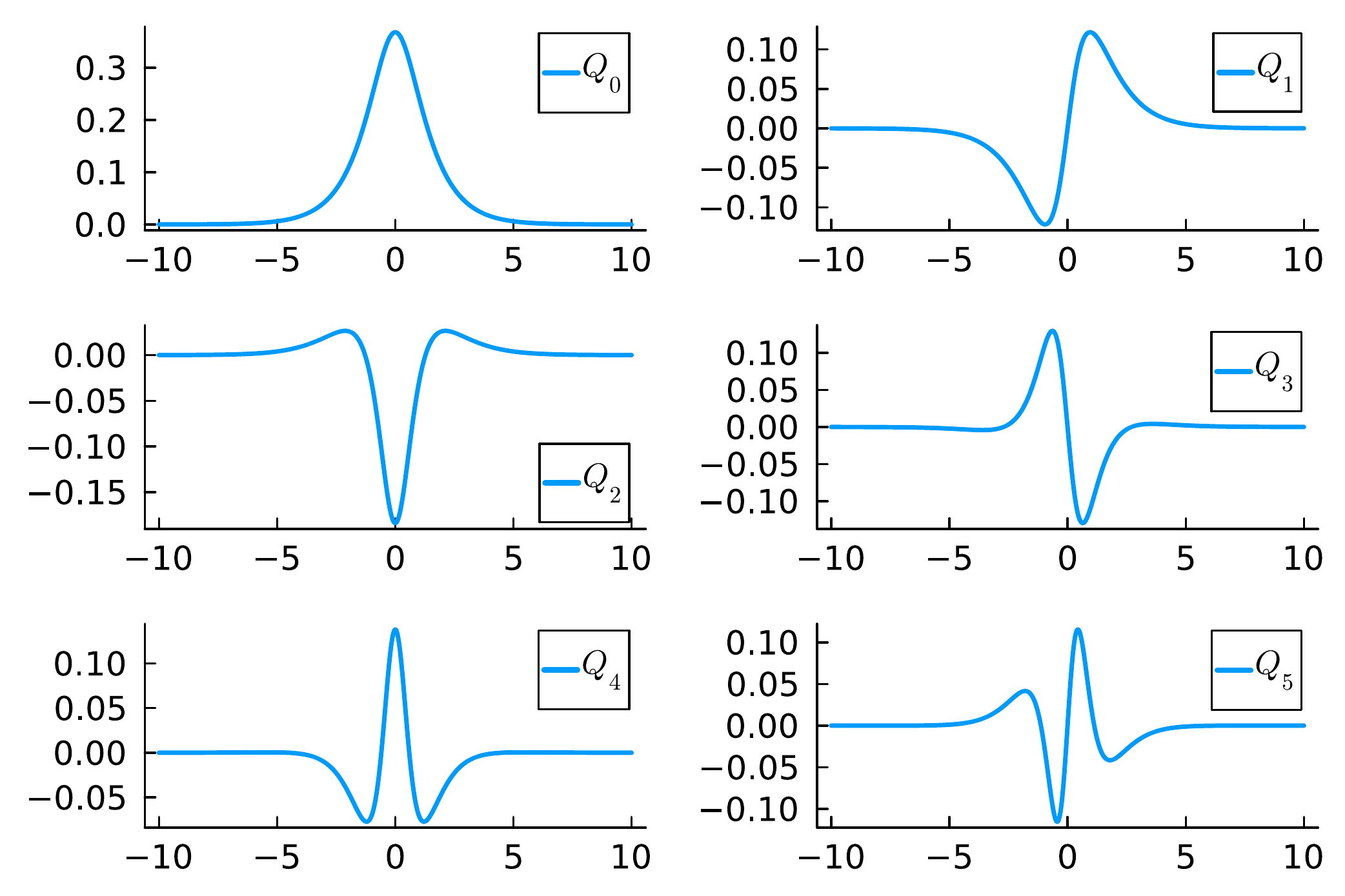}
\caption{The first six position basis functions $Q_k:=\Le_k\big(2\arctan(\cdot)/\pi\big)\cdot \exp(- \sqrt{1+r^2})$, illustrating the effect of the coordinate transform and envelope. The envelope is not part of the one-particle basis but is included here for illustration purposes. 
\label{fig:Tk}
}
\end{figure}

In practice, the global basis functions \eqref{basis} should sometimes be placed at each atomic position $R_I~(I=1,\cdots,M)$, giving rise to a new set of one-particle basis functions
\begin{align}
\label{basis_atom}
\basis_{(I,k,s)}(x) :=
\Le_k\bigg(\frac{2}{\pi}\arctan\Big(\frac{r-R_I}{L}\Big)\bigg)\cdot\chi_s(\sigma)
\quad{\rm for}~ x=(r,\sigma)\in\R\times\Z_m .
\end{align}
This construction is particularly effective when atoms are spread over a large region; see \cref{sec:numerics} for the example of Hydrogen chains when the separation distance between atoms is large.
For simplicity of presentations, we will focus on the one-particle basis set \eqref{basis} in the following, but all discussions are easily generalized to \eqref{basis_atom}.

The one-particle basis functions for the ``spin'' coordinate will play an important role in our constructions of the anti-symmetric wave functions. 
We will take $m=2$ with $\Z_2=\{\us,\ds\}$ representing the spin coordinates (see \cref{sec:vandermonde,sec:backflow}).

Finally we specify a degree for each one-particle basis function which will give us a natural approximation parameter to truncate the basis. 
Since each $\chi_s$ plays an ``equivalent'' role it is natural to associate the basis function $\basis_\ind$ with the degree of the spatial component, i.e., we define 
\begin{align}
	\label{deg}
	\deg(\basis_\ind) = \deg\big(\basis_{(k,s)}\big) := k .
\end{align}

\subsection{Sparse symmetric polynomials}
\label{sec:ace}
Our next step is to parameterize symmetric functions on $\Omega^N$, i.e., functions $f(x_1,\cdots,x_N)$, $x_j \in \Omega$, that are invariant under any permutation $\Pa \in S(N)$,
\begin{eqnarray}
	\label{symm}
	f(x_{\Pa(1)},\cdots,x_{\Pa(N)}) = f(x_1,\cdots,x_N) .
\end{eqnarray}
Without exploiting any available structure we could simply expand $f$ in terms of tensor products of the one-particle basis, i.e. 
\begin{align*}
    f \sim \sum_{\aceind \in (\N\times \Z_m)^N} c_{\aceind} \bigotimes_{t = 1}^N \phi_{\ind_t}
\end{align*}
We will perform four steps to make this approach computationally efficient even for large particle numbers $N$: (i) exploit symmetry; (ii) transform the basis into the ACE format; 
(iii) restrict the correlation order; (iv) total degree sparsification.

{\it (i) Symmetric basis: } Naively, one might construct a basis for this class by symmetrizing tensor products, 
\begin{align*}
    \phi_{\aceind}^{\rm sym} := \sum_{\Pa \in S(N)} \bigotimes_{t = 1}^N \phi_{\ind_{\Pa(t)}}
    \qquad{\rm for}~\aceind\in (\N\times \Z_m)^N_{\rm ord}. 
\end{align*}
Restricting our parameterisation to symmetric polynomials significantly reduces the number of basis functions since only ordered tuples $\aceind$ need to be considered now. 
However, due to the $O(N!)$ cost of the symmetrisation the computational cost remains the same. 

{\it (ii) The ACE formalism: } The ACE formalism constructs a basis that spans the same space as the functions $\phi_{\bm \nu}^{\rm sym}$: 
for $\aceind = (\ind_1, \dots, \ind_N) \in (\N\times \Z_m)^N_{\rm ord}$ we define  
\begin{equation}
	\label{basis:ace}
    \begin{split}
	\ace_{\aceind}(\xx) = \ace_{\ind_1, \dots, \ind_N}(x_1,\cdots,x_N)
	&:= \prod_{t=1}^N A_{\ind_t}(x_1,\cdots,x_N), \qquad \text{where} \\ 
    A_\ind(x_1,\cdots,x_N) &:= \sum_{n=1}^N \basis_\nu(x_n).
    \end{split}
\end{equation}
The $A_\ind$'s are best thought of as an generalisations of the power sum polynomials, a classical concept from invariant theory. This construction replaced the $O(N!)$ scaling cost of the naive symmetrized basis with an $O(N)$ cost for the pooling operation and an $O(1)$ cost per basis function via a recursive evaluation algorithm~\cite{bachmayr2021polynomialACE,dusson2022ACE,Kaliuzhnyi22}.
It is shown in \cite[Lemma 2.1]{bachmayr2021polynomialACE} that 
the set 
\begin{align*}
\Big\{ \ace_{\aceind} :~ \aceind\in (\N\times \Z_m)^N_{\rm ord} \Big\}
\end{align*}
forms a complete basis for the tensor product space of symmetric polynomials and the discrete space.

{\it (iii) Truncated cluster expansion: }
Our most important sparsification step is physically motivated and arises through truncating a cluster expansion. This step is analogous to the classical ANOVA expansion \cite{efron1981Anova,griebel2010Anova}. 

To that end we 
let $0<\aceB\leq N$, and expand the definition of $\ace_{\aceind}$ to index-tuples $\aceind = (\nu_1, \dots, \nu_\aceB) \in (\N\times\Z_m)_{\rm ord}^{\aceB}$
via the analogous expression, 
\begin{equation} 
\label{eq:ncorr}
    \ace_{\ind_1, \dots, \ind_\aceB}(x_1,\cdots,x_N) := \prod_{t=1}^{\aceB} A_{\ind_t}(x_1,\cdots,x_N) .
\end{equation}
These new basis functions are still functions of $x_1, \dots, x_N$ but only include products involving 
up to $\aceB$ particles at a time. We therefore call them $\aceB$-correlations.
A truncated cluster expansion for a symmetric target function $f$ can formally be written as 
\begin{align}
\label{ace:expansion}
f \sim 
\sum_{\aceind \in (\N\times\Z_m)_{\rm ord}^{\aceB}}
c_{\aceind} \ace_{\aceind}.
\end{align}
Placing an upper bound $\aceB$ on the correlation order results in a significantly lower-dimensional model. 

For future reference we denote the {\it spatial correlation order} of the basis function $\ace_{\aceind}$ by
\begin{align*}
    \|\aceind\|_0
    := \#\Big\{t \big| \deg(\basis_{\ind_t}) \neq 0\Big\}
    \qquad{\rm for}~\aceind\in (\N\times\Z_m)_{\rm ord}^{\aceB} .
\end{align*}

\begin{remark}[Cluster Expansion and Linear Dependence]
Slightly different from \eqref{ace:expansion}, and more in line with previous works on the atomic cluster expansion~\cite{drautz2019ACE,dusson2022ACE,drautz2022},
a cluster expansion can also be written as
\begin{align}
\label{ace:expansion:2}
f \sim 
\sum_{\ell=1}^{\aceB} \sum_{\aceind \in (\N\times\Z_m)_{\rm ord}^{\ell}}
c_{\aceind} \ace_{\aceind}
\qquad{\rm with}~~
\ace_{\aceind}(\pmb{x}) = \ace_{\ind_1, \dots, \ind_\ell}(\pmb{x}) := \prod_{t=1}^{\ell} A_{\ind_t}(\pmb{x})
\end{align}
where the basis functions are $\ell$-correlations with $\ell\leq\aceB$. It is important to note that neither \eqref{eq:ncorr} nor \eqref{ace:expansion:2} are subsets of the original symmetric basis \eqref{basis:ace} (the $N$-correlations). However they belong to the same space of symmetric polynomials. 
Secondly, $\ell$-correlations of different order $\ell$ in \eqref{ace:expansion:2} are {\em not} linearly independent, due to the fact that 
\begin{align}
    \label{embeding}
    1 \equiv (A_{0,\uparrow} + A_{0,\downarrow})/N.
\end{align}
Using this fact, we can write all $\ell$-correlations as a short linear combination of $n$-correlations with $n > \ell$. 
Although no longer a basis, it is straightforward to show that the $\ell$-correlations in \eqref{ace:expansion:2} satisfy a frame property, and numerical tests suggest that all our results can be reproduced with \eqref{ace:expansion:2}. For the sake of clarify of presentation we will from now on only focus on the formulation given in \eqref{ace:expansion}. 
\end{remark}

{\it (iv) Total degree sparsification: } 
The leading order term of the basis function $\ace_{\aceind}$ in \eqref{eq:ncorr} is precisely the symmetrisation of $\otimes_{t=1}^{\aceB} \phi_{\ind_t}$, and we therefore assign it the total degree 
\begin{eqnarray}
	\label{norm1}
	\deg(\ace_{\aceind}) 
	:= \|\aceind\|_1
	:= \sum_{t=1}^{\aceB} \deg(\basis_{\ind_t}) 
	\qquad{\rm for}~\aceind\in (\N\times\Z_m)_{\rm ord}^{\aceB} ,
\end{eqnarray}
where the degree $\deg(\basis_{\ind})$ of one-particle basis functions is given by \eqref{deg}. 
Truncating the many-body basis $\ace_{\aceind}$ by specifying a maximal total degree $\aceD\in\N_+$
\begin{align*}
\| \aceind \|_1 \leq \mathcal{D}, 
\end{align*}
leads to a significant sparsification of the basis that is appropriate for analytic target functions~\cite{griebel2007sparseSchrodinger,smolyak1963sparse}. It is shown in~\cite{bachmayr2021polynomialACE} (see also \eqref{eq:basissize}) that it is particularly effective when used in conjunction with symmetrisation where it significantly alleviates the curse of dimensionality of a conventional discretization using full tensor product grids.

By allowing the maximum total degree to depend on the spatial correlation order $\|\aceind\|_0$ of the basis function we can tune the required resolution for different spatial correlation orders. This is convenient since it is natural to expect that higher correlation order terms contribute less to the total energies and therefore require less resolution.
To do this, we specify a tuple $\pmb{\aceD} = (\aceD_\ell)_{\ell = 1}^{\aceB} \in \N_+^{\aceB}$.
Then the corresponding set of basis function indices is given by 
\begin{align}
\label{ace:index:set}
\aceI(\aceB,\pmb{\aceD}) := \Big\{ \aceind\in (\N\times\Z_m)_{\rm ord}^{\aceB} ~: ~~
\|\aceind\|_1 \leq \aceD_{\ell} ~{\rm if}~ \|\aceind\|_0=\ell
\Big\}, 
\end{align}
which results in a parameterisation of an $N$-variable symmetric function $f$ by
\begin{align}
	\label{ACE:approx}
	& f(\xx) \approx f_{\aceB,\pmb{\aceD}}(\xx)
	\in ~V^N_{\aceB,\pmb{\aceD}}~
	:= {\rm span} \Big\{ \ace_{\aceind} : \aceind\in \aceI(\aceB,\pmb{\aceD}) \Big\}, 
	\\[1ex]
    \notag 
    & 
    f_{\aceB,\pmb{\aceD}}(\xx)
    = \sum_{\aceind\in \aceI(\aceB,\pmb{\aceD})} c_{\aceind} \ace_{\aceind},
\end{align}
with $\params:=\{c_{\aceind}\}_{\aceind\in \aceI(\aceB,\pmb{\aceD})}$ the coefficients of the expansion.

The main result of \cite{bachmayr2021polynomialACE} estimates the number of parameters of the representation \eqref{ACE:approx} by 
\begin{equation}
\label{eq:basissize}
{\rm dim}\big(V^N_{\aceB,\pmb{\aceD}}\big)
\lesssim 
\begin{cases}
{m^{\aceB}\sum_{\ell = 1}^{\aceB} D_\ell^\ell / (\ell!)^2 } & \text{as} ~ \min_\ell D_\ell \rightarrow \infty,
\\[1ex]
{m^\aceB \sum_{\ell=1}^{\aceB}  \exp\big(\pi \sqrt{\frac{4}{3}D_\ell}\big)}
& \text{uniformly}~\text{in}~N,\aceB,\pmb{\aceD}.
\end{cases}
\end{equation}
An unsurprising but important observation is that the number of parameters depends only on the correlation order $\aceB$ and not on the number of particles $N$. This means that the curse of dimensionality only enters if a high correlation order is required to resolve the target functions. 
Even then, the permutation invariance significantly ameliorates the effect: 
In our setting we have $m = 2$ hence the second bound depends only very mildly on the correlation order $\aceB$, which suggests that even high correlation orders are easily tractable within our framework. 
Finally, we note that one can often expect that less resolution is required for higher correlation order terms, i.e. that we can take $D_\ell$ strictly decreasing and this can further alleviate the computational cost. 
For example, when computing wave functions we expect that the Hartree--Fock model, which corresponds to $\aceB = 1$, already resolves a significant contribution to the exact wave function, and subsequent corrections are relatively ``small''. 

Explaining these heuristic ideas rigorously goes far beyond the present work, but we refer to \cite{thomas2022ordACE} for a concrete and rigorous example in the context of coarse-graining electronic structure models. There, it is possible to balance all approximation parameters to obtain an explicit super-algebraic best approximation rate \cite{bachmayr2021polynomialACE}.

\subsection{Anti-symmetry via Vandermonde determinants}
\label{sec:vandermonde}
Having established our approach to parameterising high-dimensional symmetric functions we now return to the original task of parameterising an anti-symmetric many-electron wave function $\Psi(x_1,\cdots,x_N)$ with $x_i=(r_i,\sigma_i)$ having both space and spin coordinates.
The first approach we consider, which we label the {\it Vandermonde ansatz}, is based on the well-known fact that every anti-symmetric polynomial is divisible by the Vandermonde product, yielding a {\it symmetric polynomial} in the result \cite{cauchy1815vand}. 

We express a general wave function as
\begin{align}
\label{wf:vand}
\wfv(\xx) := \Phi(\xx) \cdot \mathcal{V}(\xx),
\end{align}
where $\Phi$ is symmetric with respect to the interchange of position, spin pairs, and $\mathcal{V}$ is a determinant with the ``monomial" $v_i~(1\leq i\leq N)$ depending on electron configuration $\xx$ through the number of spin up particles $N_{\xx\us}:=\#\big\{\sigma_i:\sigma_i=\us \big\}$,
\begin{align}
\label{determinant:V}
& \mathcal{V}(\xx) := \det     \begin{pmatrix}
		v_1(x_1) & \cdots & v_N(x_1)
		\\
		\vdots & \ddots & \vdots 
		\\
		v_1(x_N) & \cdots & v_N(x_N) 
\end{pmatrix} 
\\[1ex]  \nonumber
& v_i(x_j) := v_i(r_j,\sigma_j;\xx) := 
\begin{cases}
        r_j^{i-1}\cdot \delta_{\sigma_j,\us}, & i \leq N_{\xx\us}
        \\[1ex]
        r_j^{i-1 - N_{\xx\us}}\cdot \delta_{\sigma_j, \ds}, & N_{\xx\us}+1 \leq i \leq N 
\end{cases}. 
\end{align}

The symmetric prefactor $\Phi$ can be parameterized by the ACE model discussed in \cref{sec:ace}. 
For given correlation-order $\aceB\in\N$ and degrees $\pmb{\aceD}\in\N^{\aceB}$, we parameterise it as $\Phi(\xx) = f_{\aceB,\pmb{\aceD}}\big(\xx)$, where and $f_{\aceB,\pmb{\aceD}}$ is given by  \eqref{ACE:approx}.
Thus, the final form of the ACE-Vandermonde ansatz becomes
\begin{equation}
    \label{ace:vandfinal}
    \wfv(\xx) := f_{\aceB,\pmb{\aceD}}\big(\xx) \cdot \mathcal{V}(\xx).
\end{equation}

\begin{remark}[Conventional Vandermonde determinant]
If we permute the electron configurations and let the first $N_{\xx\us}$ particles spin up, 
\begin{equation*}
\xx = \big( r_1,\us; \cdots; r_{N_{\xx\us}},\us; r_{N_{\xx\us}+1},\ds; \cdots; r_N,\ds \big),
\end{equation*}
then \eqref{determinant:V} is reduced to a well-known form of Vandermonde determinant
\begin{equation}
\label{det:V:b}
\mathcal{V}(\xx) =  \prod_{\substack{1 \leq i < j \leq N \\ \sigma_i = \sigma_j = \us}} 
(r_i-r_j) \cdot 
\prod_{\substack{1 \leq i < j \leq N \\ \sigma_i = \sigma_j = \ds}} 
(r_i-r_j).
\end{equation}
Though more convenient in practice, using \eqref{det:V:b} directly in \eqref{wf:vand} will not lead to an anti-symmetric wave function, see \cref{sec:spinassign} for more discussions.
\end{remark}

\begin{remark}[Completeness]
Using similar arguments as those in \cite{cauchy1815vand,han2019dnn} one can prove that the Vandermonde ansatz \eqref{wf:vand} is ``complete''.
By this statement we mean that any anti-symmetric $N$-variable polynomial can be written in the form \eqref{wf:vand} and \eqref{ace:vandfinal}.
One can moreover show that, if an anti-symmetric function is sufficiently smooth, then it can be written exactly in the form of \eqref{wf:vand}, where the symmetric prefactor $\Phi$ is still smooth and thus can be well approximated by an ACE polynomial. We make these statements precise in \cref{theorem:vandermonde} in \cref{sec:proofs}, following the analysis in \cite[Theorem 3 and 4]{hutter2020representanti}. 

For particles in dimension greater than one, the naive Vandermonde ansatz \eqref{wf:vand} (with a suitable generalisation of the notion of Vandermonde determinant) is {\it not} complete, but instead one requires a series expansion. We refer to \cite{han2019dnn} for further details. 
\end{remark}

\begin{remark}[Numerical stability]
\label{remark:slater}
In our numerical tests with the Vandermonde ansatz, we observed numerical instabilities especially when the number of electrons is large.
This is not unexpected since the Vandermonde determinant employs monomials as matrix elements.
To resolve this, one can replace the Vandermonde determinant with a Slater determinant which leads to a generalisation of the Hartree--Fock model, or a special case of the Backflow transformation which we will explore in next subsection. 
Alternatively, a carefully designed optimization algorithm can also overcome these instabilities; here we refer to the cascadic multilevel solver we develop in \cref{sec:cascadic}. 
\end{remark}

\subsection{Anti-symmetry via backflow transformation}
\label{sec:backflow}

The backflow transformation was originally proposed by Feynman and Cohen \cite{feynman1956backflow}. It has recently seen great success when combined with deep neural network architectures~\cite{hermann2020dnn,pfau2020Ferminet,Scherbela22}. 
The backflow ansatz can be viewed as a generalization of the Slater determinant, 
\begin{align}
\label{hartree-fock}
	\Psi(x_1,\cdots,x_N) \approx \det \begin{pmatrix}
		\varphi_1(x_1) & \cdots &\varphi_N(x_1)
		\\
		\vdots & \ddots &\vdots
		\\
		\varphi_1(x_N) & \cdots & \varphi_N(x_N)
	\end{pmatrix} ,
\end{align}
where $\varphi_k~(1\leq k\leq N)$ are orthonormal one-particle orbitals.
This is the most well-known choice and most common choice to approximate the many-electron wave function, and leads to the celebrated Hartree--Fock method \cite{helgaker2014molecular}.
The idea of the backflow ansatz is to replace the one-particle orbitals $\varphi_k$ with multi-variable functions; i.e., 
\begin{align}
\label{bf:ace:1}
\wfbf(\xx) := 
	\det \begin{pmatrix}
		\varphi_1(x_1;\xx_{\neq 1}) & \cdots &\varphi_{N}(x_1;\xx_{\neq 1})\\
		\vdots & \ddots &\vdots\\
		\varphi_1(x_{N};\xx_{\neq N})  & \cdots &\varphi_{N}(x_{N};\xx_{\neq N}) 
	\end{pmatrix} 
\end{align}
where $\xx_{\neq i}:= \{ x_k \}_{k \neq i}$. The notation $\varphi_j(x_i; \xx_{\neq i})$ indicates that these generalized orbitals satisfy the ``partial'' symmetry 
\begin{align}
\label{bf:ace:3}
	\varphi_j\big(x_i;\{ x_{\Pa (k)} \}_{k \neq i} \big) 
    = \varphi_j\big(x_i; \xx_{\neq i} \big)
    \qquad \forall \Pa \in S(N-1).
\end{align}
An equivalent way to state this is to say that the vector $(\varphi_1, \dots, \varphi_N)$ is permutation-covariant. 
It is straightforward to see that the functional form \eqref{bf:ace:1} again satisfies the anti-symmetry \eqref{anti}.

To complete the ansatz \eqref{bf:ace:1}, we must specify how to parameterize the $N$-particle orbitals $\varphi_j$ in terms of the ACE formalism.
To that end, we modify the basis functions~\eqref{basis:ace}, separating out the ``highlighted'' particle $x_i$,
\begin{align*}
    \ace^{\rm BF}_{\aceind}(x_i; \{ x_k\}_{k \neq i} )
	:= \phi_{\ind_1}(x_i)\prod_{t=2}^{N} A_{\ind_t}( \{x_k\}_{k \neq i})
	\qquad{\rm for}~x_i, x_k\in\R\times\Z_2, 
\end{align*}
where $\aceind=(\ind_1; \ind_2,\cdots,\ind_N)\in (\N\times\Z_2)\times(\N\times\Z_2)_{\rm ord}^{N-1}$ and $\phi_{\ind}$ is the one-particle basis function (cf. \eqref{basis}). 
As a tensor product, it follows immediately that $\ace^{\rm BF}_{\aceind}$ form a complete basis of functions satisfying the partial symmetry \eqref{bf:ace:3}.

We now select a maximum correlation order $\aceB \geq 1$ and degrees $\pmb{\aceD} = (\aceD_\ell)_{\ell = 1}^{\aceB}$, and ACE basis for backflow orbitals, analogous to \eqref{eq:ncorr}, as
\begin{align}
    \label{basis:ace:bf:B}
    \ace^{\rm BF}_{\aceind}\big(x_i; \{ x_k\}_{k \neq i}\big)
	:= \left\{
	\begin{array}{ll}
	\phi_{\ind_1}(x_i) & {\rm if}~\aceB=1
	\\[1ex]
	\phi_{\ind_1}(x_i)\prod_{t=2}^{\aceB} A_{\ind_t}( \{x_k\}_{k \neq i}) & {\rm if}~\aceB\geq 2
	\end{array}
	\right.
\end{align}
for $x_i, x_k\in\R\times\Z_2$,
where $\aceind$ belongs to the basis index set
\begin{align*}
    \aceI^{\rm BF}(\aceB, \pmb{\aceD}) := 
    \Big\{
        \aceind \in    (\N\times\Z_2)\times(\N\times\Z_2)_{\rm ord}^{\aceB-1} ~:~~
        \|\aceind\|_1 \leq \aceD_{\ell}~{\rm if}~\|\aceind\|_0 = \ell
    \Big\}.
\end{align*}

We can now parameterise $\varphi_j$ using the basis functions in \eqref{basis:ace:bf:B}, 
\begin{equation}
	\label{ace:Phi:bf}
    \begin{split}
	 \varphi_j(x_i;\xx_{\neq i})
    &= f^{(j)}_{\aceB,\pmb{\aceD}}\big(x_i; \xx_{\neq i} \big)
    := \sum\limits_{\aceind\in \aceI^{\rm BF}(\aceB,\pmb{\aceD})} c^{(j)}_{\aceind} \ace^{\rm BF}_{\aceind}(x_i; \xx_{\neq i}).
    \end{split}
\end{equation}
Therefore, the final form of the ACE-backflow ansatz becomes
\begin{align}
\label{ace:bf:final}
\wfbf(\xx) := 
	\det \begin{pmatrix}
		f^{(1)}_{\aceB,\pmb{\aceD}}(x_1;\xx_{\neq 1}) & \cdots & f^{(N)}_{\aceB,\pmb{\aceD}}(x_1;\xx_{\neq 1})\\
		\vdots & \ddots &\vdots\\
		f^{(1)}_{\aceB,\pmb{\aceD}}(x_{N};\xx_{\neq N})  & \cdots & f^{(N)}_{\aceB,\pmb{\aceD}}(x_{N};\xx_{\neq N}) 
	\end{pmatrix} .
\end{align}
This completes our specification of the ACE-backflow wavefunction ansatz. 

\begin{remark}[Completeness]
    It is shown in \cite[Theorem 2 and 7]{hutter2020representanti} that the backflow ansatz strictly generalizes the Vandermonde determinants by absorbing the symmetric part into the determinant (see also \cref{theorem:backflow} in \cref{sec:proofs} for details).
    Therefore, the backflow ansatz gives also a ``complete'' representation of anti-symmetric polynomials for 1d systems.
    However, this may be false in higher dimensions. In \cite{huang2021geometrybackflow} it is argued that the totally anti-symmetric polynomials cannot be efficiently represented by the backflow ansatz in the category of polynomials, but that a sum over exponentially many determinants of this form is necessary.
    Note, however, that this does not prevent the backflow ansatz from having excellent approximation properties also for higher-dimensional particles. It merely highlights that understanding its approximation properties rigorously is likely challenging. 
\end{remark}

\begin{remark}[Implementation note]
    The general ACE framework leads to a convenient implementation of the backflow basis $\ace^{\rm BF}_{\aceind}$. In order to ``signal'' that exchanging $x_i, x_j$ is not a symmetry operation, one simply introduces an additional artificial spin category $\emptyset$ and modifies the particles as follows: 
    \[
        x_i' := (r_i, \emptyset), \quad x_j' := x_j \quad \text{for } j \neq i. 
    \]
    Letting $\ace_{\aceind}'$ be the standard ACE basis but now for the particle space $\R \times \Z_3$ with $\Z_3$ representing the categories $\us, \ds, \emptyset$, one can now simply take $\ace^{\rm BF}_{\aceind}(x_i, \xx_{\neq i}) := \ace_{\aceind}'(\xx')$. With this strategy, a standard ACE implementation can be employed to account for the partial symmetry. 
\end{remark}

\subsection{Spin assigned wave functions}
\label{sec:spinassign}
Since the Hamiltonian \eqref{hamiltonian:SE} is a spin-independent operator, it is convenient to work with spin-assigned wave functions \cite{foulkes2001qmc,Martin2016Interactingelectrons}.
More precisely, we consider a system with $\nus$ spin up electrons and $\nds$ spin down electrons, where $N=\nus+\nds$ gives the total number of electrons and $S_z=\frac{1}{2}(\nus-\nds)$ the spin polarization.
We can then rewrite a general wave function for such a spin-assigned system in terms of the space coordinates alone, 
\begin{align}
\label{spin-assigned}
\Psi(x_1,\cdots,x_N) = \Psi(r_1,\us,\cdots,r_{\nus},\us,r_{\nus+1},\ds,\cdots,r_N,\ds) 
=: \wfs(\rr^{\us},\rr^{\ds}),
\end{align}
where we have relabeled the electron indices so that the first $\nus$ electrons are spin up with positions $\rr^{\us}=(r_1,\cdots,r_{\nus})$ and the remaining $\nds$ electrons are spin down with positions $\rr^{\ds}=(r_{\nus+1},\cdots,r_N)$.

Then we can approximate the spin-assigned wave function by
\begin{align}
\label{wf:spin:ass}
\wfs(\rr^{\us},\rr^{\ds})
& = \Psi_{\rm V}\big(x_1, \dots, x_N) ~~{\rm or}~~ \Psi_{\rm BF}\big(x_1, \dots, x_N), 
\qquad \text{where} 
\\[1ex]
\nonumber
& \quad 
x_i = \begin{cases}
    (r_i, \us), & i = 1, \dots, \nus, \\ 
    (r_i, \ds), & i = 1, \dots, \nds, 
\end{cases} 
\end{align}
and $\Psi_{\rm V}$, $\Psi_{\rm BF}$ were constructed in \cref{sec:vandermonde}, \cref{sec:backflow} respectively. 

An important advantage of fixing the spins in the trial wave function ansatz, which we will see immediately in the next section, is that we only have to sample the configurations of space coordinates $(r_1,\cdots,r_N)\in\R^N$ in the variational Monte Carlo method.

Note that the canonical extension of $\wfs$ to $x_i$ coordinates defined through \eqref{wf:spin:ass} is already fully anti-symmetric and requires no additional anti-symmetrisation step.

\section{Cascadic multilevel optimization}
\label{sec:multilevelVMC} \setcounter{equation}{0} \setcounter{figure}{0}
Using the ACE ansatz \eqref{ace:vandfinal} or \eqref{ace:bf:final}, the many-electron wave function is parameterized as $\Psi \approx \Psi_{\param}$, where $\param\in\R^\dof$ is the collection of all ACE parameters with $\dof$ the total number of parameters. The ground state energy \eqref{E:RQ} can be approximated by minimizing a ``loss'', the Rayleigh quotient, with respect to the wave function parameters, 
\begin{align}
\label{E:parm:L}
\min_{\param\in\R^\dof} \loss(\param)
\qquad{\rm with}\qquad
\loss(\param) := E(\Psi_{\param}) = \frac{\<\Psi_{\param}|\ham|\Psi_{\param}\>}{\<\Psi_{\param}|\Psi_{\param}\>} .
\end{align}
We will discuss optimization strategies to minimize $\loss(\param)$ and design a {\it cascadic multilevel method}, inspired by \cite{braess1999cascadic,deuflhard1994cascadic,haber18a,haber18b,kimmel03,shaidurov1996cascadic}, that is efficient for our hierarchical ACE-based wave function parameterizations. 
Our multilevel approach can in principle be applied to any wave function parameterization that possesses a hierarchical structure, see e.g. \cite{Glielmo20}.

\subsection{The VMC method}
\label{sec:vmcGradient}
The evaluation of the loss $\loss(\param)$ in \eqref{E:parm:L} involves a potentially high-dimensional integral.
Within the VMC framework, this integral is typically estimated by taking samples drawn from the distribution $P(\cdot,\Psi_{\param})$ ($P$ and the local energy $E_L$ are defined in \eqref{vmc:P})
\begin{align}
\label{loss}
\loss(\param) 
\nonumber
& := E(\Psi_{\param})
= \EE_{\xx\sim P(\cdot;\Psi_{\param})} \big[ \El(\xx;\Psi_{\param}) \big]
\\[1ex]
& ~\approx \frac{1}{n}\sum_{\xx\in\sse_n} \El(\xx;\Psi_{\param}) 
=: \loss_n(\param) .
\end{align}
Here, the set $\sse_n$ for electron configurations is a set of $n$ samples $\xx \in \sse_n$, $\xx \sim P(\cdot;\Psi_{\param})$. These samples are generated using Markov-Chain Monte Carlo methods, with $n$ being the number of samplers.
Since we are working with the spin-assigned wave functions \eqref{wf:spin:ass} in practical simulations, it is only necessary for us to sample the electron configurations with space coordinates in $\R^n$ as the spin coordinates have been fixed.

In practice, we use the standard Metropolis-Hastings algorithm \cite{foulkes2001qmc}.
We have also compared this with the standard Langevin Monte Carlo approach \cite{Martin2016Interactingelectrons} and find a similar performance for sampling.

\begin{remark}[Implementation note]
We use $n=2000$ samples in practice and run 10 steps of Metropolis-Hasting sampling each time after the updates of parameters. 
The proposed moves are Gaussian distribution with an isotropic covariance chosen on the fly, such that the acceptance rates can be kept in the target range [45\%,55\%].
We maintain this ratio through a simple scheme that increases the step width by a small amount if the acceptance ratio strays too far above the target, and decreases it by a small amount if the ratio strays too far below the target.
\end{remark}

The optimization of minimizing $\loss(\param)$ can be achieved by a stochastic gradient descent (SGD) iteration
\begin{align}
\label{sgd}
\param_{k+1} = \param_k - \eta_k \cdot \grad_n(\param_k) ,
\end{align}
where $\eta_k>0$ is the learning rate
and $\grad_n(\param)$ is an approximation of the gradient $\nabla_{\param}\loss(\param)$  estimated by using a set of samples $\sse_n$ from the distribution $P(\cdot,\Psi_{\param})$
\begin{align}
\label{gradient}
\nonumber
\nabla_{\param}\loss(\param)
& = 2~\EE_{\xx\sim P(\cdot;\Psi_{\param})} \big[ \big(\partial_{\param}\log|\Psi_{\param}(\xx)|\big) \big(\El(\xx;\Psi_{\param}) - \loss(\param)\big) \big]
\\[1ex]
& \approx \frac{2}{n}\sum_{\xx\in\sse_n} \big(\partial_{\param}\log|\Psi_{\param}(\xx)|\big) \big(\El(\xx;\Psi_{\param}) - \loss_n(\param)\big)
=: \grad_n(\param) .
\end{align}
This formulation gives an unbiased estimate of the gradient of the loss.
The derivation of the gradient expression can be found in \cite{lin2021explicitly}.
We refer to \cite{abra23,liwenchen23} for two recent works on the convergence analysis of the SGD iteration \eqref{sgd}.

In \cref{algorithm-sgd} we recall the standard VMC algorithm for solving \eqref{E:parm:L}.

\begin{algorithm}[H]
	\caption{~~ Standard VMC algorithm}
    \label{algorithm-sgd}
	\vskip 0.1cm
	\hspace*{0.02in} 
	{\bf Input:} 
	wave function parameterisation $\Psi_{\param}$; initial guess $\params_0\in\R^{\dof}$; iterations $T\in\Z_+$.
	\begin{algorithmic}[1]
	  \For{$k=0, 1, \dots, T$}
	  \State
	  Generate a set $S_n$ of $n$ samples $\xx \in \sse_n$,  $\xx \sim P(\cdot,\Psi_{\param_k})$ (e.g., via MCMC).
	  \State
	  Use \eqref{gradient} to approximate gradients. 
      \State 
      Use \eqref{sgd} to update the parameters, $\params_k\mapsto\params_{k+1}$.
      \EndFor
	\end{algorithmic}
	\hspace*{0.02in} 
	{\bf Output:} 
	$\params = \params_T$ and the corresponding energy $\loss(\params)$.
\end{algorithm}

Different approaches have been developed to accelerate the convergence of the optimization.
For example, the so-called stochastic reconfiguration method (or equivalently the natural gradient method) \cite{sorella1998sr,sorella2001Lanczos} preconditions the stochastic gradient by the inverse of the Fisher information matrix; 
the linear methods \cite{toulouse2007linearmethod} involves the second order derivatives of the loss; 
and the Kronecker Factorized Approximate Curvature (KFAC) method \cite{martens2015kfac} uses a Kronecker-factored approximation of the Fisher information matrix.
We are not going to explore these techniques in this work. 
Instead, we apply a weighted Adam (AdamW) scheme \cite{Loshchilov2018ADAMW} and develop a cascadic multilevel method (see next subsection) to accelerate the convergence of  \cref{algorithm-sgd}.
We use a decaying learning rate $\eta_k=\alpha/(1+k/\beta)$ with fixed parameters $\alpha$ and $\beta$.

\subsection{The cascadic multilevel method}
\label{sec:cascadic}
When the ACE-based wave function parameterisations \eqref{ace:vandfinal} and \eqref{ace:bf:final} are employed with relatively large correlation order and polynomial degrees, the standard VMC algorithm (\cref{algorithm-sgd}) has a slow convergence to the ground state due to the ill-conditioning in the loss function.
We propose a cascadic multilevel algorithm that exploits the intrinsic hierarchy of the ACE model and accelerates the optimization process significantly.
The cascadic multilevel method was originally proposed in \cite{bornemann1996cascadic} for second-order elliptic problems, but it is natural to apply the idea to many other hierarchical parameterisations.
We refer to \cite{braess1999cascadic,deuflhard1994cascadic,shaidurov1996cascadic} for more discussions and applications of this multilevel method. 
The idea has been explored in training neural networks, using more iterations at ``coarser scales'' to obtain good starting guesses for the ``finer scales'' and avoid being trapped in local minima \cite{haber18a,haber18b,kimmel03}.
Similar ideas have also been applied to optimizing Gaussian process state wave functions for quantum many-body problems \cite{Glielmo20}, which derive a linear combination of all possible multi-site correlation features in the implicit parametrization in the form of an ``additive'' kernel, and hierarchically includes all lower-rank features.

We remark that our method is not to be confused with ``multilevel Monte Carlo'' (MLMC) methods \cite{giles15}. 
The idea of the MLMC methods is, when computing expectations in stochastic simulations, to take most samples at a low accuracy with a relatively low cost, and take only very few samples at high accuracy with a correspondingly high cost.
MLMC methods also exploit model hierarchies, but its goal is to accelerate the sampling whereas our multilevel method is to accelerate the parameter optimization but using a naive sampler at each step.
MLMC methods cannot be easily incorporated into our VMC algorithm since the distribution from which we draw samples is different at each level.    

For a given electron number $N$, let 
\begin{align}
	\label{multilevel:spaces}
	X_0\subsetneq X_1\subsetneq \cdots\subsetneq X_j\subsetneq \cdots \subsetneq X_L
\end{align}
be a sequence of (strictly) nested multilevel function classes for many-electron wave functions. At each level $j$ an element of $X_j$ is written as $\Psi_j(\cdot; \param^{(j)})$ where the parameters $\param^{(j)}$ are unconstrained in a vector space $\mathbb{R}^{\mathcal{M}_j}$. 

The nestedness of the parameterisations \eqref{multilevel:spaces} implies that there exists a prolongation operator 
\begin{equation}
    \Pi_{j, j+1} : \mathbb{R}^{\mathcal{M}_j} \to \mathbb{R}^{\mathcal{M}_{j+1}}, 
    \qquad \param_0^{(j+1)} := \Pi_{j, j+1}[ \param^{(j)} ]
\end{equation}
such that 
\begin{equation}
    \Psi_{j+1}\big(\cdot; \param_0^{(j+1)}\big)
        \equiv 
    \Psi_{j}\big(\cdot; \param^{(j)}\big).
\end{equation}

The idea of the cascadic multilevel method is to first apply the VMC algorithm at a computationally inexpensive coarse level $X_0$ resulting in parameters $\param^{(0)}$, using the prolongation $\Pi_{0,1} \param^{(0)}$ as initial guess for another VMC iteration at level $X_1$, and then to iterate the procedure. This is expressed in the following algorithm:

\begin{algorithm}[H]
	\caption{~~ Cascadic multilevel VMC algorithm}
	\label{algorithm-multilevel}
	\vskip 0.1cm
	\hspace*{0.02in} 
	{\bf Input:} levels $\{X_j\}_{j=0}^L$; iteration steps $\{M_j\}_{j=0}^L$;
	initial guess $\params_0^{(0)}$. 
	\begin{algorithmic}[1]
		\For{$j=0, 1, \dots, L$}
		\State 
		$\params_{M_j}^{(j)} \leftarrow $ Algorithm \ref{algorithm-sgd} with parameterisation $X_j$, initial guess $\params_0^{(j)}$, $M_j$ steps.
		\State
		Set $\params_0^{(j+1)} = \prj_{j,j+1} \params_{M_j}^{(j)}$ .
		\EndFor
	\end{algorithmic}
	\hspace*{0.02in} 
	{\bf Output:} 
	$\params=\params_{M_L}^{(L)}$ and the corresponding energy $\loss(\params)$.
\end{algorithm}

\begin{remark} The multilevel approach essentially handles the low-frequency and high-frequency components of the target wave function separately at different steps, so the same acceleration can be achieved as in the classical multilevel methods. A distinctive feature of the cascadic multilevel method is that the coarse function class can be discarded once they are refined, which considerably simplifies it algorithmically.
\end{remark}

\subsection{Hierarchy and prolongation for ACE wave functions}
It remains to specify the function class hierarchy and prolongation operator for the ACE-Vander-monde and ACE-Backflow wave function parameterisations that we introduced in the previous sections. 

Let $\{\aceB^{(j)}\}$ and $\{\pmb{\aceD}^{(j)}\}$ be
``increasing" truncation parameters satisfying
\begin{align*}
\aceB^{(j)}\leq\aceB^{(j+1)}
\qquad{\rm and}\qquad
\aceD^{(j)}_{\ell} \leq \aceD^{(j+1)}_{\ell}
\quad{\rm for}~\ell=1,\cdots,\aceB^{(j)} ,
\end{align*}
where at least one of the above `$\leq$' should be `$<$'.
Then we can specify the multilevel function classes $X_j$ as
\begin{align*}
    X_j := \Big\{ \Psi=\wfv :~ \wfv~
	{\rm can~be~written~as~\eqref{ace:vandfinal}~with~} \aceB=\aceB^{(j)}~{\rm and}~\pmb{\aceD}=\pmb{\aceD}^{(j)} 
	\Big\}
\end{align*}
if the wave functions are parameterized by ACE-Vandermonde ansatz, and
\begin{align*}
    X_j := \Big\{ \Psi=\wfbf :~ \wfbf~
	{\rm can~be~written~as~\eqref{ace:bf:final}~with~} \aceB=\aceB^{(j)}~{\rm and}~\pmb{\aceD}=\pmb{\aceD}^{(j)} 
	\Big\}
\end{align*}
if the wave functions are parameterized by ACE-backflow~ansatz.

When climbing from a low level $X_j$ to a higher level $X_{j+1}$, we include basis functions that have either higher polynomial degrees or higher correlation orders, so that some higher resolution components (higher frequency or correlation) are added to the basis set.
The corresponding prolongation operators $\varPi_{j,j+1}$ can be given explicitly in the following.

Let $\param^{(j)}$ be the parameters on the $j$-th level, then to embed it into a finer level as $\param^{(j+1)}$, we should define the prolongation $\varPi_{j,j+1}$ from the $j$-th to the $(j+1)$-th level.
The definition of prolongation is given for two cases: (1) $\aceB^{j+1} = \aceB^{j}$; and (2) $\aceB^{j+1} = \aceB^{j} + 1$.
Note that the case $\aceB^{j+1} - \aceB^{j} > 1$ can be obtained by the compositions of case (2).

For the first case $\aceB^{j+1} = \aceB^{j}$, we define $\varPi_{j,j+1}$ by 
\begin{align}
\label{prolongation:I}
& \param^{(j+1)} = \varPi_{j,j+1} \param^{(j)}
\qquad {\rm with}
\\[1ex]
& \qquad
\Big( \varPi_{j,j+1} \param^{(j)} \Big)_{\aceind}
:= \left\{
\begin{array}{ll}
	\param^{(j)}_{\aceind} & ~ {\rm if}~ \aceind\in \aceI(\aceB^{(j)},\pmb{\aceD}^{(j)})
	\\[1ex] \nonumber
	0 & ~ {\rm if}~ \aceind\in \aceI(\aceB^{(j+1)},\pmb{\aceD}^{(j+1)}) \backslash \aceI(\aceB^{(j)},\pmb{\aceD}^{(j)})
	\end{array}
	\right.  .
	\qquad
\end{align}

For the second case $\aceB^{j+1} = \aceB^{j} + 1$, since the index set $\aceI(\aceB^{(j)},\pmb{\aceD}^{(j)})$ is not a subset of $\aceI(\aceB^{(j+1)},\pmb{\aceD}^{(j+1)})$, we need first introduce two embeddings for the indices 
\begin{align*}
	& \eb_{\us}: \aceI(\aceB,\pmb{\aceD})\rightarrow \aceI(\aceB+1,\pmb{\aceD}), 
	\qquad 
	\eb_{\us}\big(\ind_1, \cdots, \ind_{\aceB}) = \Pa\Big( (0, \us), \ind_1, \cdots, \ind_{\aceB} \Big)
	\qquad{\rm and}
	\\[1ex]
	& \eb_{\ds}: \aceI(\aceB,\pmb{\aceD})\rightarrow \aceI(\aceB+1,\pmb{\aceD}), 
	\qquad  
	\eb_{\ds}\big(\ind_1, \cdots, \ind_{\aceB}) = \Pa\Big( (0, \ds), \ind_1, \cdots, \ind_{\aceB} \Big),
\end{align*}
where an additional permutation $\Pa$ is needed in the definition such that $\eb_{\us}(\aceind), \eb_{\ds}(\aceind) \in (\N\times \Z_m)_{{\rm ord}}^{\aceB+1}$. 
Then the prolongation $\varPi_{j,j+1}$ is defined by 
\begin{align}
\label{prolongation:II}
& \param^{(j+1)} = \varPi_{j,j+1} \param^{(j)} \qquad {\rm with}
\\[1ex] \nonumber
& \qquad 
\Big( \varPi_{j,j+1} \param^{(j)} \Big)_{\aceind}
:= 
\left\{
\begin{array}{ll}
    \displaystyle
	\frac{1}{N} \sum_{\aceind^{(j)}\in\aceI_{\aceind}} \params^{(j)}_{\aceind^{(j)}}
	& {\rm if}~ \aceI_{\aceind} \neq\emptyset
	\\[1ex] 
	0 & {\rm otherwise}
	\end{array}
	\right.  
\end{align}
for $\aceind\in \aceI\big(\aceB^{(j+1)},\pmb{\aceD}^{(j+1)}\big)$, where
\begin{align*}
\aceI_{\aceind} := \Big\{ \aceind^{(j)} \in \aceI\big(\aceB^{(j)},\pmb{\aceD}^{(j)}\big) ~:~ 
\eb_{\us}\big(\aceind^{(j)}\big) = \aceind ~ {\rm or}~ \eb_{\ds}\big(\aceind^{(j)}\big) = \aceind
\Big\} .
\end{align*}
By using \eqref{embeding}, we see that the prolongation $\varPi_{j,j+1} \param^{(j)}$ gives
\begin{align*}
& \Psi_{\param^{(j+1)}} 
= \sum_{\genfrac{}{}{0pt}{}{\aceind\in\aceI(\aceB^{(j+1)},\pmb{\aceD}^{(j+1)})}{\aceI_{\aceind} \neq\emptyset}} \frac{1}{N} \sum_{\aceind^{(j)}\in\aceI_{\aceind}} \params^{(j)}_{\aceind^{(j)}} \ace_{\aceind}
\\[1ex]
&\qquad
= \sum_{\aceind^{(j)}\in\aceI(\aceB^{(j)},\pmb{\aceD}^{(j)})} \param^{(j)}_{\aceind^{(j)}}\ace_{\aceind^{(j)}} 
\frac{1}{N}(A_{0,\uparrow} + A_{0,\downarrow})
= \sum_{\aceind^{(j)}\in\aceI(\aceB^{(j)},\pmb{\aceD}^{(j)})} \param^{(j)}_{\aceind^{(j)}}\ace_{\aceind^{(j)}}
= \Psi_{\param^{(j)}} ,
\end{align*}
which implies that the wave functions are exactly the same before and after the prolongation.

We see that the prolongations defined by both \eqref{prolongation:I} and \eqref{prolongation:II} can be easily and efficiently implemented in practice.

\section{Numerical experiments}
\label{sec:numerics} \setcounter{equation}{0} \setcounter{figure}{0}
In this section, we demonstrate the efficiency of our algorithm by simulations of several representative example systems.
All simulation results are given in atomic units (a.u.).
All simulations are performed on a workstation with 16 Intel Xeon W-3275M processors and 1T RAM, using the {\sc Julia} \cite{Julia} package {\tt ACESchrodinger.jl} \cite{ACESchrodinger}.

We will perform simulations for one-dimensional versions of second-row atoms, 
the molecule Lithium-Hydrogen (LiH), and Hydrogen chains (${\rm H}_n$).
To illustrate the accuracy of our results, we make comparisons with those from other standard quantum chemistry methods: (i) restricted Hartree--Fock (RHF), (ii) unrestricted Hartree--Fock (UHF) and (iii) Jastrow-Slater (JS).
In the Jastrow-Slater calculations, the wave functions are parameterized as a product of a Slater determinant of the form \eqref{hartree-fock} and a Jastrow factor.
The benchmarking Jastrow factors we used have the following form:
\begin{align*}
	& J(r_1,\sigma_1, \cdots, r_N,\sigma_N) 
	= \sum_{i=1}^N \chi(r_i)
	+ \sum_{i<j}\omega^{\sigma_i\sigma_j}\big(|r_i-r_j|\big)
	\qquad\qquad{\rm with}
	\\
	& \chi(r) = \sum_{I = 1}^M \xi_I\sqrt{1+|r-R_I|^2}
	\quad {\rm and} \quad
	\omega^{\sigma\zeta}(u) = \frac{a^{\sigma\zeta}}{\sqrt{1+u^2}} \big(1-\exp(-b^{\sigma\zeta}\sqrt{1+u^2}) \big) ,
\end{align*}
where $\xi_I~(I=1,\cdots,M)$ and $a^{\sigma\zeta},~b^{\sigma\zeta}~(\sigma,\zeta\in\{\us,\ds\})$ are tunable parameters.

To visualize the ground state solutions, we show the single-electron density $\rho(r)$ and pair-electron density $\rho_2(r,r')$, which are defined from the many-electron wave functions by
\begin{align*}
	\label{density}
	\rho(r) &= N \sum_{\sigma_1,\cdots,\sigma_N\in\Z_2}\int \big| \Psi(r,\sigma_1,r_2,\sigma_2, \dots,r_N,\sigma_N) \big|^2 {\rm d}r_2\dots ~{\rm d}r_N, \quad \text{and}  
	\\[1ex]
	\rho_2(r,r') &= \binom{N}{2}\sum_{\sigma_1,\cdots,\sigma_N\in\Z_2}\int \big| \Psi(r,\sigma_1,r',\sigma_2,r_3,\sigma_3,\dots,r_N,\sigma_N) \big|^2 {\rm d}r_3\dots ~{\rm d}r_N .
\end{align*}
The integrals in the above definition are evaluated by Monte-Carlo methods with Metropolis-Hastings sampling.

\vskip 0.2cm

\noindent{\bf Example 1: Atoms.}
We perform simulations for the 1d atoms in the second row of periodic table.
We first use the backflow model and compare the convergence of our cascadic multilevel algorithm with the pure AdamW method in \cref{fig:vmc:be}.
By taking the final correlation order $\aceB=2$ and polynomial degrees $\pmb{\aceD}=(16,16)$, it is observed that the multilevel method achieves significantly faster convergence than the direct AdamW approach. 
For the Beryllium atom ($N=4$), we can reach the energy -6.784 a.u. (with accuracy to 0.001 a.u.)
in 700 steps with the multilevel method, but about 1000 steps with the pure AdamW;
for the Oxygen atom ($N=8$), the multilevel method takes about 600 steps to reach the energy -21.692 a.u. (with accuracy to 0.005 a.u.), while the pure AdamW takes more than 1400 steps to achieve the same accuracy. We find that the improvement of the multilevel algorithm becomes increasingly significant as we increase the number of degrees of freedom for parameterisation.
For example, if we simulate the Neon atom with the Vandermonde model with $\aceB=3$,
a direct AdamW algorithm will fail to converge most of time, while the multilevel method can always converge to the ground state well; see \cref{table:vand:expl}. 

Moreover, we emphasize that the multilevel method not only requires fewer iteration steps to attain required accuracy, but also saves computational cost at early iterations when the wave functions are computationally much cheaper to evaluate.

\begin{figure}[tbhp]
	\centering
	\includegraphics[width=6.0cm]{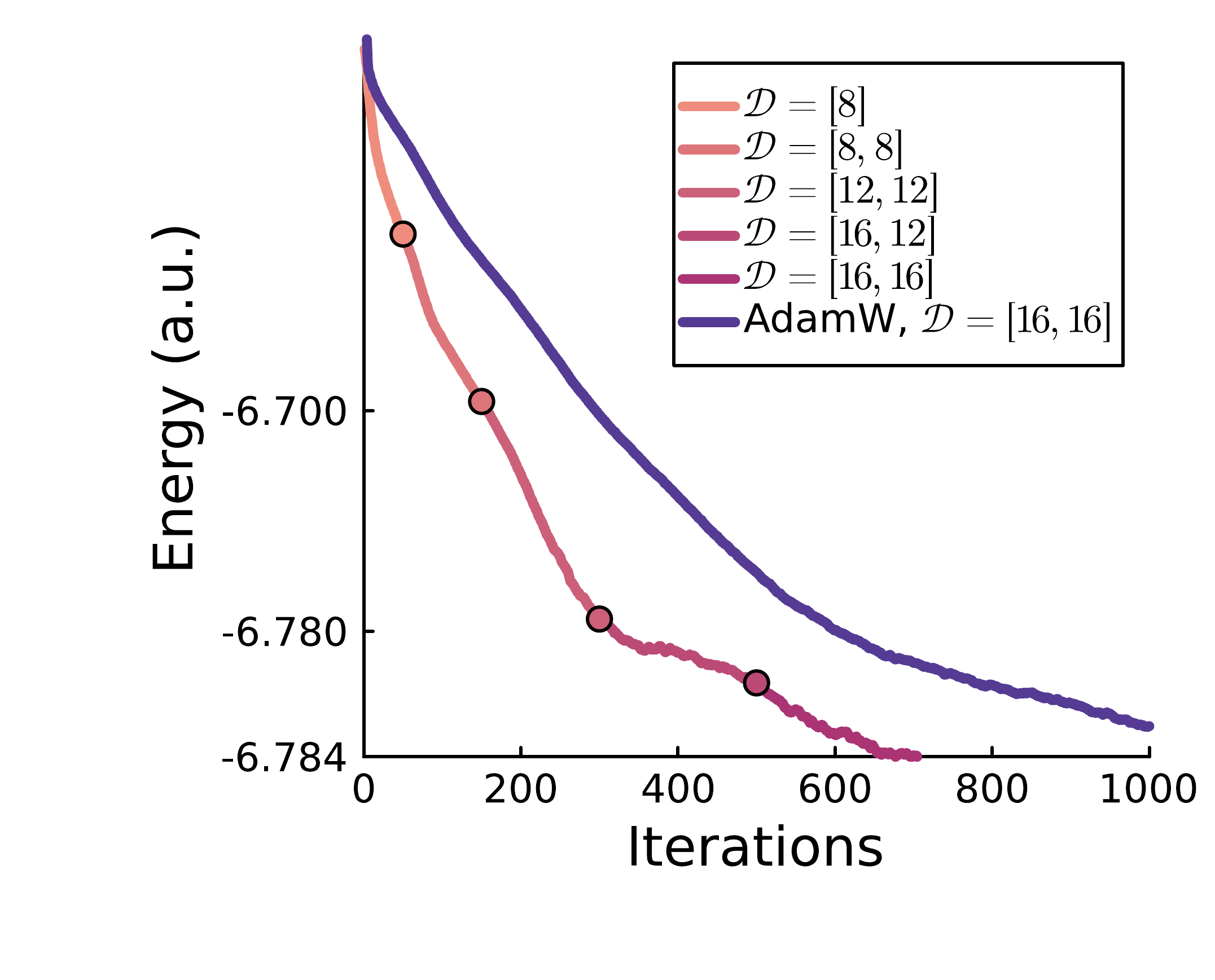}
	\includegraphics[width=6.0cm]{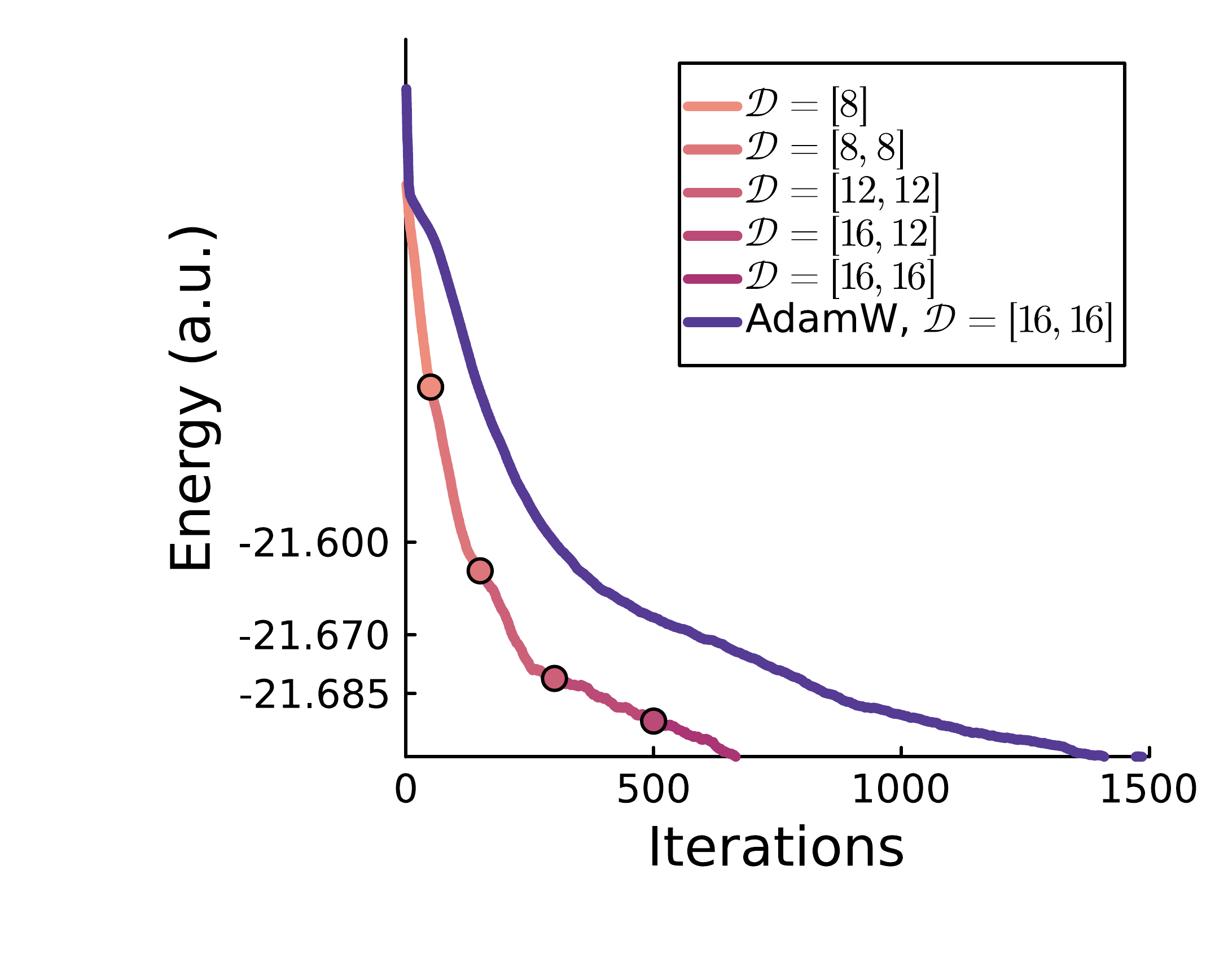}
	\caption{Convergence of the multilevel VMC algorithms with backflow model. 
		Left: Be; Right: O.}
	\label{fig:vmc:be}
\end{figure}

\begin{table}[tbhp]
	\label{table:vand:expl}
	\caption{Decay of the energy errors of the multilevel and pure AdamW methods with the Vandermonde model of Neon.
		The converged ground state energy is taken as the reference to calculate the errors.
	}
	\centering
	\begin{tabular}{cccccc}
		\toprule
		\multirow{2}*{Iteration} & \multicolumn{2}{c}{$\pmb{\aceD} = (32,16,8)$} & \multicolumn{2}{c}{$\pmb{\aceD} = (16,8,4)$}\\ 
		\cmidrule(lr){2-3} \cmidrule(lr){4-5}
		& AdamW (a.u.) & multilevel (a.u.) & AdamW (a.u.) & multilevel (a.u.)\\
		\midrule
		10 & 46.9158 & 23.0249 & 151.4131 & 17.3581\\
		20 & 1173.2703 & 21.3572 & 2368.7069 & 6.3844\\
		40 & 5172.9311 & 16.7712 & 6136.1001 & 5.0375\\
		80 & 4426.9050 & 9.9026 & 1542.2801 & 1.2489\\
		160 & 1351.9313 & 2.0020 & 3068.2906 & 0.1859\\
		320 & 38.2721 & 0.1626 & 12.1771 & 0.0395\\
		\bottomrule
	\end{tabular}
\end{table}

Next, we compare the accuracy of different parameterisations. In \cref{fig:err:be} we present the ground state energies with respect to the number of degrees of freedom. 
We observe that when reaching the same low energy, the backflow model requires more parameters than the Vandermonde model. On the other hand, the backflow ansatz needs lower correlation order and significantly smaller polynomial degrees to achieve comparable accuracy. Moreover, we observe that the backflow model achieves the same accuracy as the UHF model when the correlation order is 1, and outperforms the Jastrow-Slater model when the correlation order is 2.

\begin{figure}[tbhp]
	\centering
	\includegraphics[width=6.0cm]{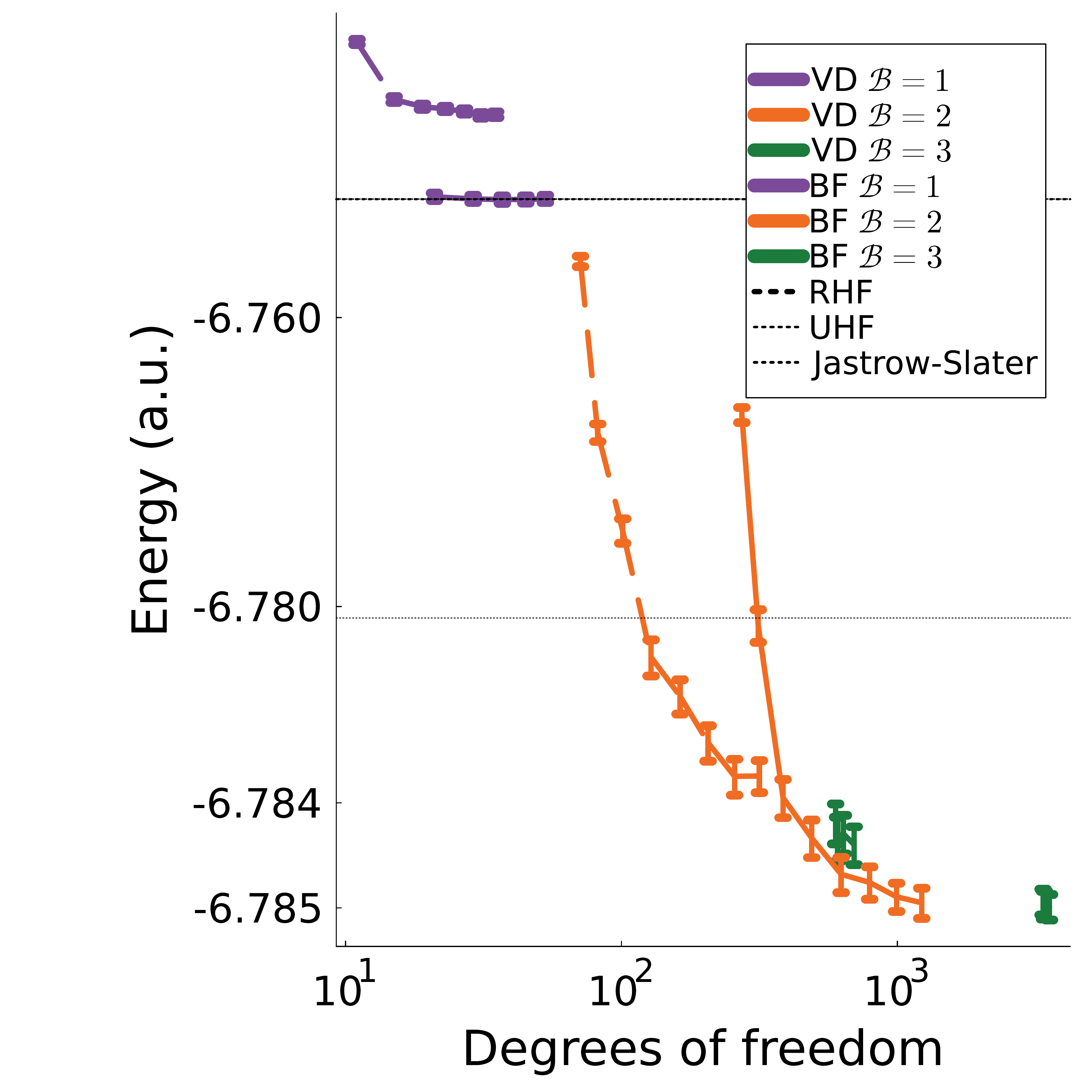}
	\includegraphics[width=6.0cm]{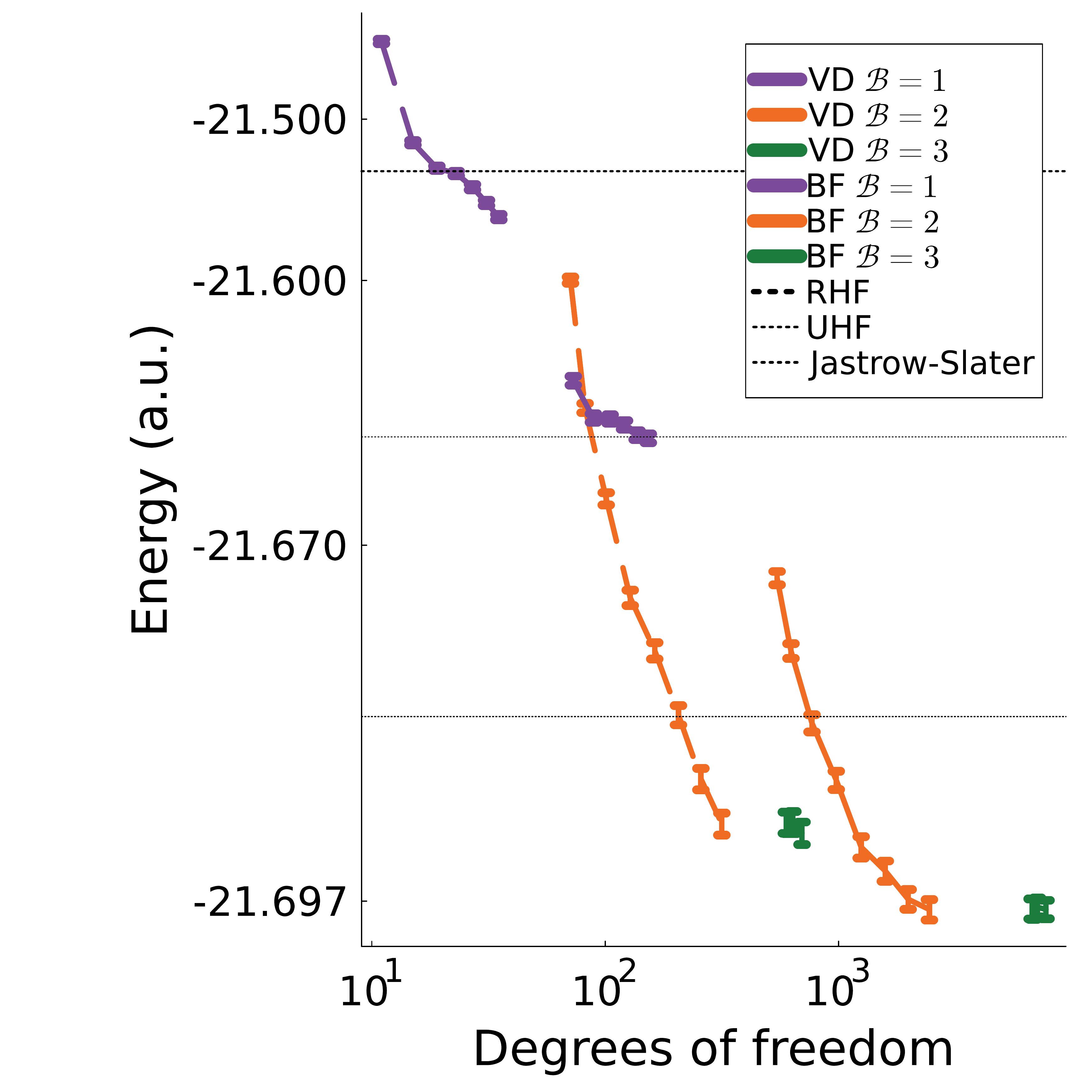}
	\caption{Convergence of the ground state energy with respect to the ACE degrees of freedom. 
		In the legends, ``BF" stands for the backflow model and ``VD" stands for the Vandermonde model.
		The error bars show the standard deviations in the reported energy.
		Left: Be; Right: O.
	}
	\label{fig:err:be}
\end{figure}

\begin{remark}
	Benchmarking the performance of a prototype code is not usually informative, hence we do not focus too much on that aspect. Nevertheless to substantiate our previous comment, 
	we present in \cref{table:timing} some timings of evaluating the wave function $\Psi_{\param}(\xx)$ and the operation of Hamiltonian $\big(\ham\Psi_{\param}\big)(\xx)$ at at given configuration $\xx$ for the Oxygen atom.
	This performance is tested on a MacBook Air with M2 chip.
	We make a further brief remark on the comparison between two models: using $\pmb{\aceD}=(16,8)$ for the Backflow ansatz (with $977$ parameters) and $\pmb{\aceD}=(16,14)$ for the Vandermonde ansatz (with $257$ parameters) the two methods achieved similar accuracy. The Backflow ansatz required $1.13$ times longer than the Vandermonde ansatz per optimization epoch. 
\end{remark}

\begin{table}[tbhp]
	\label{table:timing}
	\caption{Timing results ($\mu$s, serial) for evaluating $\Psi_{\param}(\xx)$ and $\big(\ham\Psi_{\param}\big)(\xx)$ at a given configuration $\xx$ of Oxygen.
		The corresponding degrees of freedom (dof) for each parameterisation are also given.
	}
	\centering
	\begin{tabularx}{\textwidth}{>{\centering\arraybackslash}X>{\centering\arraybackslash}c>{\centering\arraybackslash}X>{\centering\arraybackslash}X>{\centering\arraybackslash}c>{\centering\arraybackslash}X>{\centering\arraybackslash}X}
		\toprule
		\multirow{2}*{Degrees $\pmb{\aceD}$} & \multicolumn{3}{c}{Vandermonde-ACE} & \multicolumn{3}{c}{Backflow-ACE}\\ 
		\cmidrule(lr){2-4} \cmidrule(lr){5-7}
		& \# dofs & $\Psi_{\param}(\xx)$ & $\big(\ham\Psi_{\param}\big)(\xx)$ & \# dofs & $\Psi_{\param}(\xx)$ & $\big(\ham\Psi_{\param}\big)(\xx)$\\
		\midrule
		(32) & 67 & 2.53 & 85.0 & 265 & 7.41 &  81.5   \\
		(32,16) & 380 & 2.88 & 129 & 2961 & 16.9 & 152 \\
		(32,16,8) & 793 & 3.69 & 225 & 8633 & 28.6 & 276\\
		{\small (32,16,8,4)} & 1211 & 5.09 & 358 & 15137 & 41.7 & 411 \\
		\bottomrule
	\end{tabularx}
	\vskip 0.5cm
\end{table}

We further present in \cref{fig:density} the ground state single- and pair-electron densities.
We see that for this type of atom systems, the densities obtained by the backflow model are qualitatively similar to those by the RHF model.
However, there are clear differences in the single-electron densities for atoms Oxygen and Neon, as the backflow model can capture features beyond the Hartree--Fock model.

\begin{figure}[tbhp]
	\centering
	\includegraphics[width=12cm]{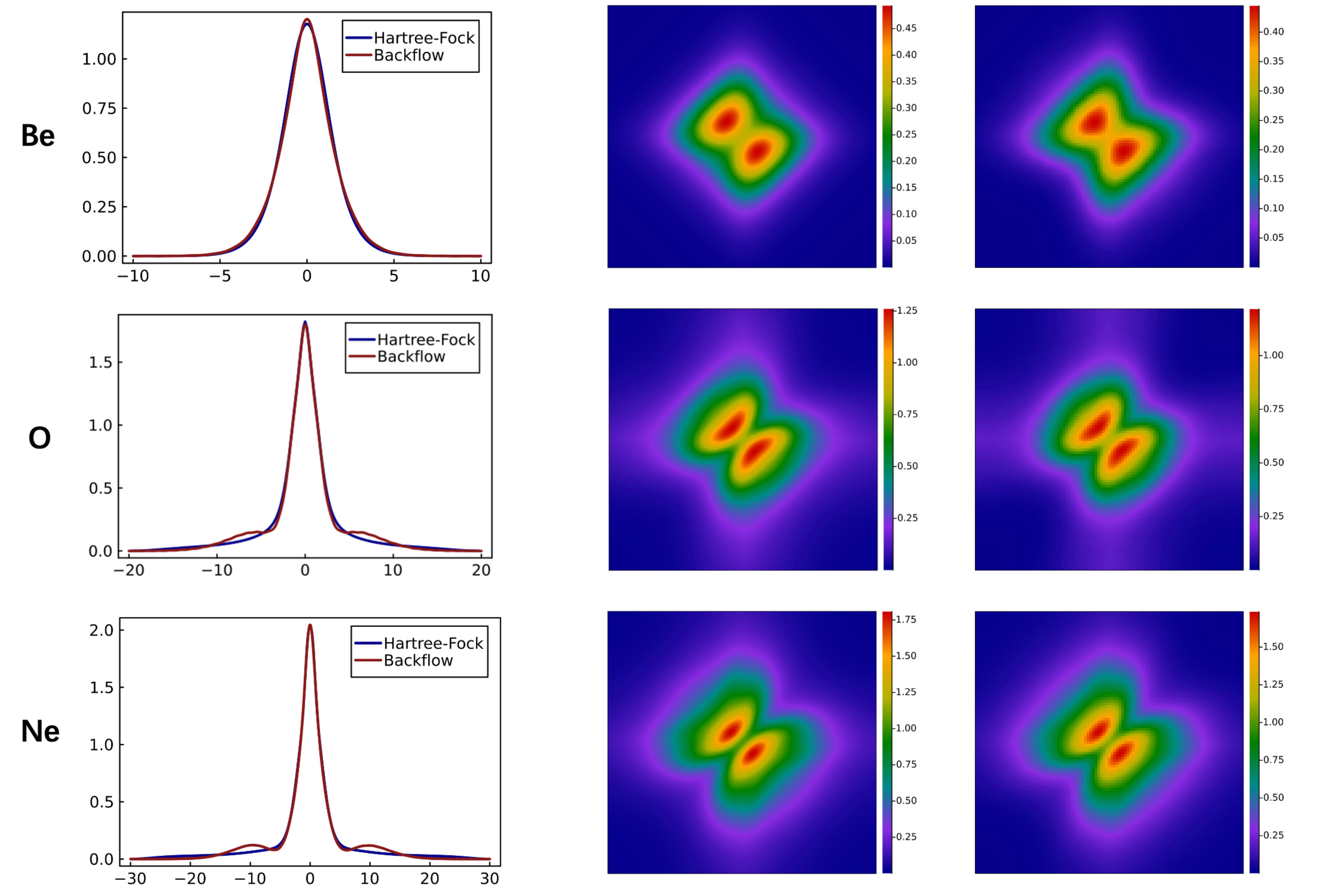}
	\caption{The ground state single- and pair- densities for the 1d atoms.
		Left: single-electron density; 
		Middle: pair-density by RHF;
		Right: pair-density by backflow.
	}
	\label{fig:density}
\end{figure}

\vskip 0.2cm

\noindent
{\bf Example 2: LiH.}
The LiH molecular has $M=2$ atoms (with nuclear charge 3 and 1 respectively) and $N=4$ electrons.
The two atoms are located at $R_1 = -d/2$ and $R_2 = d/2$ with $d=2.8$ a.u. the distance between the Lithium and Hydrogen atom. 
We present the ground state single- and pair-electron densities in \cref{fig:density:LiH}.
We see a clear left-right correlation from the picture of pair-electron density: when one electron is to the left of the origin, the probability of finding another electron favors the region on the right and vice versa; and the probability of finding two electrons to the right of the origin at the same time is very low.
\begin{figure}[tbhp]
	\centering
	\includegraphics[width=12.0cm]{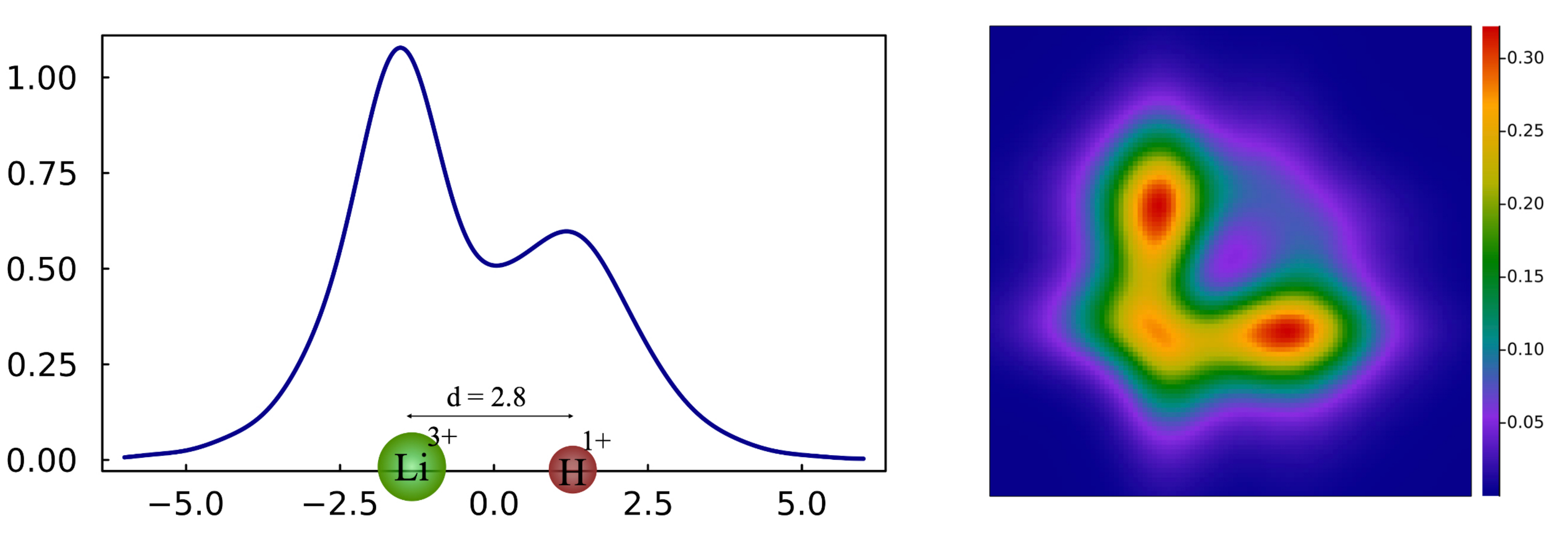}
	\caption{Left: The atomic configuration of (1d) LiH and the ground state single-electron density. Right: The ground state pair-electron density.}
	\label{fig:density:LiH}
\end{figure}

\vskip 0.2cm

\noindent
{\bf Example 3: Hydrogen chains.}
Finally, we consider Hydrogen chains $\text{H}_{n}$ with $n$ atoms and $n$ electrons.
The Hydrogen atoms are located at $R_k = - (n-1) d/2 - (k-1) d ~(k=1,2,\cdots,n)$, where $d$ stands for the separation distance (a.u.) between the atoms. 
We are interested in the dissociation limit of the chain, as the separation distance $d$ becomes large.
In that limit, the ground state of the system is expected to behave as $n$ isolated Hydrogen atoms.
Therefore, we use the ground state energy of a single Hydrogen atom, $E_1=-0.669777$ a.u. (obtained by solving the differential equation on a fine grid), as a reference value to check the accuracy of the simulation.
We mention that this is a typical strongly correlated system for which many approaches will fail to capture the right behavior \cite{cohen2012challenges,friesecke2022stronginteraction,wagner2012reference}.
The usual understanding from a quantum chemistry perspective of related problems is that of static correlation, which corresponds to an inherently multideterminental situation.
Of course there is a possibility of breaking the spin symmetry, e.g. by using the unrestricted Hartree--Fock model, which may be allowable at infinite separation, but this solution does not give the correct ground state for any other distance.

We will employ only the backflow model in this example.
We plot the dissociation curve of $H_8$ in \cref{fig:H8}, for both RHF and backflow models.
We observe that the backflow model gives an accurate ground state at the dissociation limit (by comparing with the reference $E_1$), while the RHF model displays qualitatively incorrect behavior as $d$ grows.
The ground state electron densities are also compared on the picture, from which one can again observe a clear qualitative deviation at large separation distances.
We show in \cref{table:energy} more simulation results with different atom number $n$ and separation distance $d$, and observe that the ACE-backflow model can achieve sufficient accuracy even at large separation distance $d$.
We mention that the atom-centered one-particle basis functions \eqref{basis_atom} are used for these simulations as the separation distance is quite large.

\begin{figure}[tbhp]
	\centering
	\includegraphics[width=12.0cm]{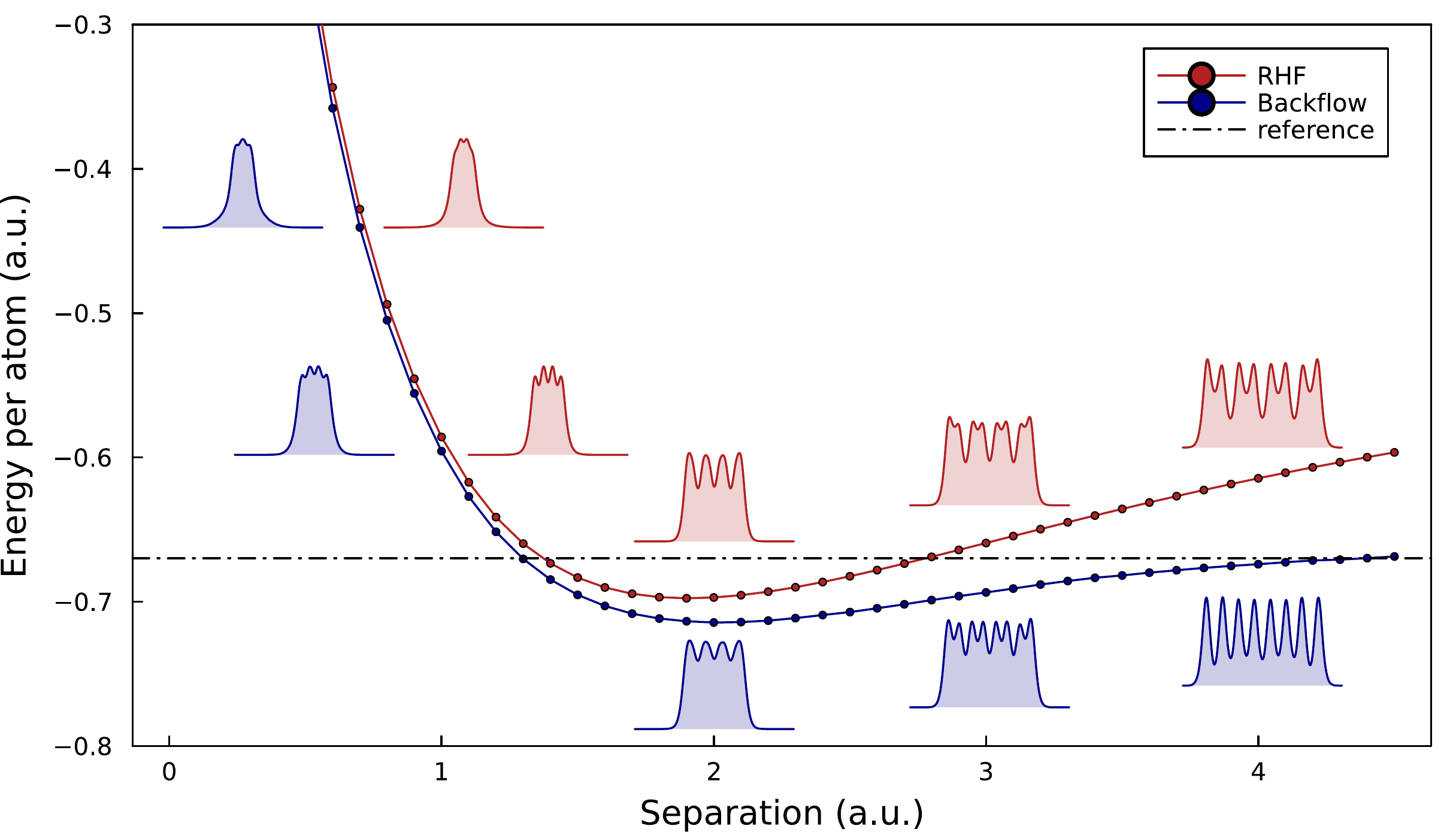}
	\caption{The dissociation curves of $H_{\text{8}}$. 
		The reference is given by the ground state energy of a single Hydrogen atom.
	}
	\label{fig:H8}
\end{figure}

\begin{table}[tbhp]
	\caption{Ground state energy (per atom) for Hydrogen chains with $n$ atoms.}
	\label{table:energy}  
	\centering 
	\vskip 0.2cm
	\begin{tabularx}{13cm}{cccccc}   
		\toprule
		Separation (a.u.) & ${\rm H}_2$ (a.u.) & ${\rm H}_4$ (a.u.) & ${\rm H}_6$ (a.u.) & ${\rm H}_8$ (a.u.) & ${\rm H}_{10}$ (a.u.) 
		\\
		\midrule
		12.0 & -0.6697(6) & -0.6697(5) & -0.6697(3) &-0.6697(1)  &-0.6697(1)
		\\
		20.0 & -0.6697(7) & -0.6697(6) & -0.6697(7) & -0.6697(7) & -0.6697(7)
		\\
		\bottomrule
	\end{tabularx} 
\end{table}

\section{Conclusions}
\label{sec:conclusion}
We developed a cascadic multilevel VMC method for solving many-electron Schr\"{o}dinger equations based on anti-symmetric ACE architectures that come with a natural multilevel structure.
We demonstrate numerically in typical one-dimensional electron systems that our approach has good performance and a systematically improvable accuracy.
It remains to demonstrate that our framework is equally performant in realistic three-dimensional systems with standard Coulomb interactions, which will require more careful construction of the one-particle basis employed in the ACE wave function architectures.
It will be particularly interesting to explore how for strongly correlated systems the performance and accuracy of our ACE parameterisations can be maintained.

\appendix
\section{Completeness of the parameterisations}
\label{sec:proofs}
In this appendix, we present analytic results to support the ``completeness'' of our parameterisations of anti-symmetric functions. The proofs will also highlight that the backflow ansatz is much more general than the Vandermonde ansatz.
Our analysis conceptually follows \cite{cauchy1815vand,han2019dnn,hutter2020representanti}. 
For the sake of simplicity, we will only discuss the spin-assign wave function $\wfs$ in \eqref{wf:spin:ass}. The generalisation to those with spin variables is straightforward. 
We emphasize that the spin assigned wave function $\wfs : \R^N\rightarrow \R$ has a different anti-symmetric requirement since it is only a function of spatial variables.
The anti-symmetric constraint for the full wave function $\Psi(\xx)$ is reduced to the following condition for for $\wfs(\rr)$
\begin{equation}
\label{eq:vand_anti-symmetric}
	\wfs(\rr^{\us},\rr^{\ds}) = (-1)^{\Pa^{\us}}(-1)^{\Pa^{\ds}} \wfs\big(\Pa^{\us}(\rr^{\us}),\Pa^{\ds}(\rr^{\ds})\big)
	\qquad\forall~\Pa^{\us},\Pa^{\us} ,
\end{equation}
where $\Pa^{\us}$ and $\Pa^{\ds}$ represent permutations of $\rr^{\us}$ and $\rr^{\ds}$, respectively.

\begin{theorem}
\label{theorem:vandermonde}
If $\wfs:\R^N\rightarrow\R$ satisfies \eqref{eq:vand_anti-symmetric}, then it can be expressed as
\begin{equation}
\label{eq:vand_spin-assign}
\wfs(\rr^{\us},\rr^{\ds}) = \Phi_{\rm s}(\rr^{\us}, \rr^{\ds})\cdot \mathcal{V}_{\rm s}(\rr^{\us},\rr^{\ds}),
\end{equation}
where $\mathcal{V}_{\rm s}(\rr^{\us},\rr^{\ds}) = \prod\limits_{1 \leq i < j \leq \nus} 
(r_i-r_j) \cdot \prod\limits_{\nus+1 \leq i < j \leq N} (r_i-r_j)$ and $\Phi_{\rm s}$ satisfies
\begin{equation}
\label{eq:ace-anti-symmetric}
\Phi_{\rm s}(\rr^{\us},\rr^{\ds}) = \Phi_{\rm s}\big(\Pa^{\us}(\rr^{\us}),\Pa^{\ds}(\rr^{\ds})\big)
\qquad\forall~\Pa^{\us},\Pa^{\ds} .
\end{equation}
Moreover, if $\wfs$ is a polynomial or analytic function, then so is $\Phi_{\rm s}$.
\end{theorem}

\begin{proof}
It is evident that $\mathcal{V}_{\rm s}$
satisfies the same anti-symmetry condition as $\wfs$. Therefore, as a quotient of two anti-symmetric functions, $\Phi_{\rm s}$ satisfies \eqref{eq:ace-anti-symmetric}.

Furthermore, if $\wfs$ is a polynomial satisfying \eqref{eq:vand_anti-symmetric}, we can use the factor theorem to show that $r^{\sigma}_i-r^{\sigma}_j~(\sigma \in{\uparrow, \downarrow})$ divides $\wfs(\rr^{\us},\rr^{\ds})$.
By the unique factorization property of multivariate polynomials, we can write $\wfs$ as
\eqref{eq:vand_spin-assign}, 
where $\Phi_{\rm s}(\rr^{\us},\rr^{\ds})$ is a polynomial with ``partial'' symmetry \eqref{eq:ace-anti-symmetric}.
Moreover, we can estimate the degree of $\Phi_{\rm s}$ by observing that the product of all $r^{\sigma}_i - r^{\sigma}_j~(\sigma \in{\uparrow, \downarrow})$ has degree $N_{\sigma}(N_{\sigma}-1)/2$. Therefore, the degree of $\Phi_{\rm s}(\rr^{\us},\rr^{\ds})$ is no more than ${\rm deg}(\wfs)-\nus(\nus-1)/2-\nds(\nds-1)/2$.
	
Next we show that, if $\wfs$ is analytic, then $\Phi_s$ is also analytic: Assuming $\wfs$ is analytic, we can write its multivariate Taylor series expansion as 
\begin{equation*}
	\wfs(\rr) = \sum\limits_{\mathbf{k}} c_{\mathbf{k}} \rr^{\mathbf{k}}, 
\end{equation*}
where $\mathbf{k} = (k_1, \ldots, k_N) \in \mathbb{N}^{N}_0$,  $\rr^{\mathbf{k}} := r_1^{k_1}\cdots r_N^{k_N}$, and $c_{\mathbf{k}} \in \R$ are the coefficients of the expansion. By anti-symmetrizing both sides using anti-symmetrization operator $AS$, 
\begin{equation*}
AS[\rr^{\mathbf{k}}] := 
\sum_{\Pa^{\us}}\sum_{\Pa^{\ds}}
(-1)^{\epsilon_{\Pa^{\us}}}  (-1)^{\epsilon_{\Pa^{\ds}}} 
\Pa^{\us}(r_1)^{k_1} \cdots \Pa^{\us}(r_{\nus})^{k_{\nus}}  \Pa^{\ds}(r_{\nus+1})^{k_{\nus}+1} \cdots \Pa^{\ds}(r_N)^{N}, 
\end{equation*}
we get $AS[\wfs(\rr)] = \sum_{\mathbf{k}} c_{\mathbf{k}} AS[\rr^{\mathbf{k}}]$. Using the result above, we know that $AS[\rr^{\mathbf{k}}]$ can be represented as 
\begin{equation*}
AS[\rr^{\mathbf{k}}] = U_{\mathbf{k}}(\rr)\cdot
\mathcal{V}_{\rm s}(\rr^{\us}, \rr^{\ds}),
\end{equation*}
where $U_{\mathbf{k}}(\rr)$ is a ``partial'' symmetric polynomial. Substituting this expression into the anti-symmetrized Taylor series expansion, we get 
\begin{equation*}
AS[\wfs(\rr)] = \sum_{\mathbf{k}} c_{\mathbf{k}} U_{\mathbf{k}}(\rr) \cdot
\mathcal{V}_{\rm s}(\rr^{\us}, \rr^{\ds}).
\end{equation*}
This is of the form given in \eqref{eq:vand_spin-assign}, with $\Phi(\rr) := \sum_{\mathbf{k}} c_{\mathbf{k}} U_{\mathbf{k}}(\rr)$, which is analytic, since $U_{\mathbf{k}}(\rr)$ is a ``partial" symmetric polynomial. 
\end{proof}

\begin{remark}[Differentiable functions]
If $\wfs:\R^N\rightarrow\R$ is $p$-times differentiable, i.e. $\wfs\in C^p(\R^N)$, then dividing $\wfs$ by the factor $r^{\us}_i-r^{\us}_j~(1\leq i<j\leq \nus)$ gives rise to a $N$-variable function in $C^{p-1}$.
We can recursively divide out $r_\ell - r_k$ and obtain a function in $C^{P-\nus(\nus-1)/2}(\R^N)$.
By repeating the same argument for the spin-down part, we can conclude that $\Phi\in C^k(\R^N)$ with $k = P-\nus(\nus-1)/2 - \nds(\nds-1)/2$.
\end{remark}

\begin{theorem}
\label{theorem:backflow}
Let $\wfs:\R^N\rightarrow\R$ satisfy \eqref{eq:vand_anti-symmetric}.
Then there exist orbitals $\varphi^i_s:\R^N\rightarrow\R~(1\leq i\leq N)$ satisfying
\begin{equation}
\label{eq: bf: symmetry}
\varphi^i_{\rm s}(r_j; \rr^{\us}_{\neq j}, \rr^{\ds}_{\neq j}) = \varphi^i_{\rm s}\Big(r_j; \Pa^{\us}(\rr^{\us}_{\neq j}), \Pa^{\ds}(\rr^{\ds}_{\neq j})\Big) 
\qquad\forall~\Pa^{\us},\Pa^{\us} , 
\end{equation}  
such that $\wfs$ and can be expressed in the form of 
\begin{align}
\label{eq:bf-spin-assign}
& \wfs(\rr) = \det \big(\Phi^{\us}(\rr)\big) \cdot \det\big(\Phi^{\ds}(\rr)\big) \qquad {\rm with}
\\[1ex] \nonumber
\Phi^{\sigma}(\rr) & = 
\begin{pmatrix}
\varphi^{1+\delta}_{\rm s}(r_{1+\delta};\rr_{\neq 1+\delta}) & \cdots & \varphi^{N_{\sigma}+\delta}_{\rm s}(r_{1+\delta};\rr_{\neq 1+\delta}) \\
\vdots & \ddots &\vdots\\
\varphi^{1+\delta}_{\rm s}(r_{N_{\sigma}+\delta};\rr_{\neq N_{\sigma}+\delta})  & \cdots &\varphi^{N_{\sigma}+\delta}_{\rm s}(r_{N_{\sigma}+\delta};\rr_{\neq N_{\sigma}+\delta}) 
\end{pmatrix}, 
\quad 
\sigma\in\{\us,\ds\} ,
\end{align}
where $\delta=0$ if $\sigma=\us$, and $\delta=\nus$ if $\sigma=\ds$.
Moreover, if $\wfs$ is a polynomial or analytic function, then so are the orbitals $\varphi^i_{\rm s}
~(1\leq i\leq N)$.
\end{theorem}

\begin{proof}
We will demonstrate that the Vandermonde ansatz \eqref{eq:vand_spin-assign} can be expressed in the form of the backflow parameterization \eqref{eq:bf-spin-assign}.
Recall that Theorem \ref{theorem:vandermonde} gives us an expression for $\wfs$ as follows:
\begin{align*}
	\wfs(\rr^{\us},\rr^{\ds}) = \Phi_{\rm s}(\rr^{\us}, \rr^{\ds}) \cdot \prod_{1\leq i<j\leq \nus}(r_i-r_j) \cdot \prod_{\nus+1\leq k<\ell\leq N}(r_k-r_\ell),
\end{align*}
where $\Phi_{\rm s}(\rr^{\us}, \rr^{\ds})$ satisfies \eqref{eq:ace-anti-symmetric}.
We can then rewrite the above equation in the form
\begin{align*}
\wfs(\rr^{\us},\rr^{\ds}) &= \left|
\begin{matrix}
  1 & \Phi_{\rm s}(\rr^{\us}, \rr^{\ds})\cdot r_1&\cdots & r_1^{\nus-1}\\
   \vdots & \vdots & \ddots & \vdots\\
1 & \Phi_{\rm s}(\rr^{\us}, \rr^{\ds})\cdot r_{\nus}&\cdots & r_{\nus}^{\nus-1}\\
\end{matrix}
\right| \cdot \left|
\begin{matrix}
    1 & r_{\nus+1}&\cdots & r_{\nus+1}^{\nds-1}\\
    \vdots & \vdots & \ddots & \vdots\\
    1 & r_{N}&\cdots & r_{N}^{\nds-1}\\
\end{matrix}
\right|.
\end{align*}
Consider functions $\varphi^i_{\rm s}(r_j;\rr^{\us}_{\neq j},\rr^{\ds}_{\neq j})$, defined as follows, 
\begin{align*}
\varphi^i_{\rm s}(r_j;\rr^{\us}_{\neq j},\rr^{\ds}_{\neq j}) = 
\begin{cases}
r_j\cdot \Phi(\rr^{\us}, \rr^{\ds}), &i=2\\
 r_j^{i-1}, &i\leq \nus, i\neq 2\\
r_j^{i-1-\nus}, & i>\nus
\end{cases} ,
\end{align*}
It is clear that functions $\varphi^i_{\rm s}$ satisfy a form of ``partial" symmetry, as defined in equation \eqref{eq: bf: symmetry}.
Moreover, if $\wfs$ is a polynomial, an analytical function, then $\varphi^i_{\rm s}$ possesses corresponding properties such as being a polynomial, or an analytical function. 
\end{proof}

\section{Laplacian implementation}
\label{sec:laplacian-implementation}
A key algorithmic component in the VMC algorithms is the efficient evaluation of the laplacian operator, followed by differentiation with respect to parameters to obtain the gradient of the loss. This involves three derivatives in total, hence requires a brief remark on its implementation. 

For the sake of simplicity we focus this section on the Backflow parameterisation, but the comments are generally applicable. Let $\xx \in (\R \times \Z_m)^N$ be the electrons with space and spin coordinate. Each forward evaluation of the wave function can be understood as a chain
\[
\begin{array}{rcccccccccl}
    (\R \times \Z_m)^N 
    &\rightarrow&  \R^{N \times K} 
    &\rightarrow& \R^{N \times K_1} &\rightarrow& \R^{N \times K_2} &\rightarrow& \R^{N \times N} &\rightarrow& \R. \\ 
    \xx &\mapsto& \le_K &\mapsto& A &\mapsto& \ace &\mapsto& \Phi &\mapsto& \Psi,
\end{array} 
\]
where $K, K_1, K_2$ denote, respectively, the number of basis functions in $L_K, A$ and ${\bm A}$. $\le_K$, $A$ and $\ace$ correspond to the arrays storing the basis function for the $N$ inputs $(x_i; \xx_{\neq i})$. The orbitals $\Phi$ are obtained by a matrix multiplication, $\Phi = C \cdot {\bm A}$ with $C$ the matrix of orbital parameters, and the wave function is then computed as $\log |\Psi| = \log{|\det{\Phi}|}$ for numerical stability.

The gradient $\nabla_{\xx}\Psi = \nabla \Psi$ can be calculated efficiently by backward differentiation, but for evaluating $\Delta \Psi$ we are unaware of a similar technique and 
can only make use of the forward-mode differentiation. Naively, this might involve computing the entire hessian $\nabla_{\xx}^2 \Psi$ as a forward-mode differentiation of $\nabla_{\xx}\Psi$ followed by taking the trace. Instead, we observe that it can be computed in a single forward pass storing and computing only minimal information, which appears to be more efficient in our implementation. However, this efficiency gain is only in the prefactor when compared to the naive approach of computing the full hessian.

We start by observing that 
\[
\Delta \Psi = \Phi^{-\trans} : \Delta \Phi - \sum_{i = 1}^N (\Phi^{-1} \frac{\partial \Phi}{\partial x_i})^{\trans} : (\Phi^{-1} \frac{\partial \Phi}{\partial x_i}),
\]
where $:$ is the matrix dot product. Thus, to obtain $\Delta_{\xx} \Psi$ we require $\Delta_{\xx} \Phi$ and the Jacobian $\nabla_{\xx} \Phi$. Arguing recursively, a single forward pass can be used to evaluate $\Psi, \nabla\Psi, \Delta\Psi$, resulting in the following chain of operations: 
\begin{equation*}
\xx \mapsto (\le, \partial_x \le_K, \partial_x^2 \le_K) \mapsto (A, \nabla A, \Delta A) \mapsto (\ace, \nabla \ace, \Delta \ace) \mapsto (\Phi, \nabla \Phi, \Delta \Phi) \mapsto \Delta\Psi.
\end{equation*}
The key observation is that the full hessians are never evaluated or stored, but this comes at a cost of computing the Jacobians in forward-mode.

The observations made above result in an efficient evaluating of the loss \eqref{loss}. Finally, to compute its gradient with respect to the parameters, we can now use highly efficient backward-mode differentiation. 

\section*{Acknowledgments}
We thank Zeno Sch\"{a}tzle and Juerong Feng for inspiring conversations on the topic of this article.


\bibliographystyle{siamplain}
\bibliography{main.bib}

\end{document}